%% file: main.tex
\documentclass[acmsmall]{ec22acm}

\AtBeginDocument{%
  \providecommand\BibTeX{{%
    \normalfont B\kern-0.5em{\scshape i\kern-0.25em b}\kern-0.8em\TeX}}}

\copyrightyear{2022}
\acmYear{2022}
\setcopyright{acmcopyright}\acmConference[EC '22]{Proceedings of the 23rd ACM Conference on Economics and Computation}{July      11--15, 2022}{Boulder, CO, USA}
\acmBooktitle{Proceedings of the 23rd ACM Conference on Economics and Computation (EC '22), July      11--15, 2022, Boulder, CO, USA}
\acmPrice{15.00}
\acmDOI{10.1145/3490486.3538287}
\acmISBN{978-1-4503-9150-4/22/07}

\usepackage{amsmath}
\usepackage[shortlabels]{enumitem}

\usepackage{amsthm}
\usepackage{amsfonts}
\usepackage{graphicx}
\graphicspath{{./fig/}}

\usepackage{caption}
\usepackage{subcaption}
\usepackage{float}
\usepackage{tikz}
\usetikzlibrary{bayesnet}
\usepackage[textsize=tiny]{todonotes}

\usepackage{subfiles}
\usepackage{amsmath}
\usepackage{accents}
\usepackage{bm}

\makeatletter
\if@todonotes@disabled

\else

\fi
\makeatother

\newtheorem{theorem}{Theorem}
\newtheorem*{theorem*}{Theorem}
\newtheorem{corollary}{Corollary}
\newtheorem{definition}{Definition}

\newtheorem{lemma}{Lemma}

\DeclareMathOperator{\norm}{\mathcal{N}}

\DeclareMathOperator{\E}{\mathrm{E}}
\DeclareMathOperator{\Pb}{\mathrm{P}}
\DeclareMathOperator{\R}{\mathbf{R}}

\newcommand\lims{\lim_{\s\to\infty}}
\newcommand\BR{\beta}
\newcommand\br{b}

\newcommand\F{F}

\newcommand{\A}{H}
\newcommand{\B}{L}
\newcommand{\g}{G}
\newcommand{\idx}{i}

\newcommand{\p}{p}
\newcommand{\pa}{\p_{\A}}
\newcommand{\pb}{\p_{\B}}
\newcommand{\pg}{\p_\g}

\newcommand\payoff{u}
\let\vv\relax
\newcommand\vv{\payoff}

\newcommand\vp{\payoff_\g}

\newcommand\pxaun{\pxa^\mathrm{un}}
\newcommand\pxbun{\pxb^\mathrm{un}}
\newcommand\pxgun{\pxg^\mathrm{un}}

\newcommand\mpx{\bar \px}

\newcommand\mpxaun{\mpx_\A^\mathrm{un}}
\newcommand\mpxbun{\mpx_\B^\mathrm{un}}
\newcommand\mpxgun{\mpx_\g^\mathrm{un}}

\newcommand\mpxadp{\mpx_\A^\mathrm{dp}}
\newcommand\mpxbdp{\mpx_\B^\mathrm{dp}}

\newcommand{\w}{W}

\newcommand{\wg}{\w_\g}

\newcommand{\wh}{\hat \w}

\newcommand{\wt}{\tilde \w}

\newcommand{\wtg}{\wt_{\g}}

\newcommand{\til}{\theta}
\newcommand{\tiln}{T}
\newcommand{\tild}{\theta^\mathrm{d}}
\newcommand{\tildi}{\tild_\g}
\newcommand{\tilda}{\tild_\A}
\newcommand{\tildb}{\tild_\B}
\newcommand{\tildg}{\tild_\g}

\newcommand{\tilun}{\theta^\mathrm{un}}
\newcommand{\tilg}{\til_\g}

\newcommand{\tilgdp}{\tilg^\mathrm{dp}}

\newcommand{\rev}{Q^\mathrm{un}}
\newcommand{\revdp}{Q^\mathrm{dp}}

\newcommand{\sx}{\hat\sigma}
\newcommand{\sxa}{\sx_{\A}}
\newcommand{\sxb}{\sx_{\B}}
\newcommand{\sxp}{\sx_{\g}}

\newcommand{\st}{{\tilde\sigma}}
\newcommand{\sta}{\st_{\A}}
\newcommand{\stb}{\st_{\B}}
\newcommand{\stp}{\st_{\g}}

\newcommand{\cc}{C}

\newcommand{\cp}{\cc_{\g}}
\newcommand{\ca}{\cc_{\A}}
\newcommand{\cb}{\cc_{\B}}
\newcommand{\s}{S}
\newcommand{\gun}{\mathcal{\g}^\mathrm{un}}
\newcommand{\gdp}{\mathcal{\g}^\mathrm{dp}}

\newcommand{\sqp}{\sq}

\newcommand{\m}{m}

\newcommand{\mqa}{\mu_{\A}}
\newcommand{\mqb}{\mu_{\B}}
\newcommand{\mqg}{\mu_{\g}}

\newcommand{\mmin}{\br^{\min}}
\newcommand{\mmax}{\br^{\max}}

\newcommand{\mqamax}{\mmax_{\A}}
\newcommand{\mqamin}{\mmin_{\A}}
\newcommand{\mqbmax}{\mmax_{\B}}
\newcommand{\mqbmin}{\mmin_{\B}}

\newcommand{\mup}{\mu_{\g}}

\newcommand{\mvecun}{\mvec^\mathrm{un}}

\newcommand{\mqaun}{\mqa^\mathrm{un}}
\newcommand{\mqbun}{\mqb^\mathrm{un}}
\newcommand{\mqgun}{\mqg^\mathrm{un}}

\newcommand{\mun}{\m^\mathrm{un}}
\newcommand{\maun}{\mun_{\A}}
\newcommand{\mbun}{\mun_{\B}}
\newcommand{\mpun}{\mun_{\g}}

\newcommand{\mmq}{\bar\mq}
\newcommand{\mmqaun}{\mmq_\A^\mathrm{un}}
\newcommand{\mmqbun}{\mmq_\B^\mathrm{un}}
\newcommand{\mmqgun}{\mmq_\g^\mathrm{un}}

\newcommand{\mqadp}{\mq_\A^\mathrm{dp}}
\newcommand{\mqbdp}{\mq_\B^\mathrm{dp}}
\newcommand{\mqpdp}{\mq_\g^\mathrm{dp}}
\newcommand{\mmqadp}{\bar\mq_\A^\mathrm{dp}}
\newcommand{\mmqbdp}{\bar\mq_\B^\mathrm{dp}}

\newcommand{\mqbr}{\br}
\newcommand{\mqabr}{\mqbr_\A}
\newcommand{\mqbbr}{\mqbr_\B}

\newcommand{\mq}{\mu}
\newcommand{\mvec}{\bm{\mq}}
\newcommand{\sq}{\eta}
\newcommand{\sign}{\mathrm{sgn}}

\newcommand{\px}{x}
\newcommand{\pxmax}{{x}^\mathrm{max}}
\newcommand{\pxmin}{{x}^\mathrm{min}}
\newcommand{\pxg}{\px_{\g}}

\newcommand{\pxa}{\px_{\A}}
\newcommand{\pxamax}{\pxmax_{\A}}
\newcommand{\pxamin}{\pxmin_{\A}}

\newcommand{\pxbmax}{\pxmax_{\B}}
\newcommand{\pxbmin}{\pxmin_{\B}}

\newcommand{\pxb}{\px_{\B}}

\newcommand{\ax}{\alpha}

\begin{document}
\fancyhead{}

\title{Fairness in Selection Problems with Strategic Candidates}

\author{Vitalii Emelianov}
\affiliation{
  \institution{Univ. Grenoble Alpes, Inria, CNRS, Grenoble INP, LIG}
  \city{Grenoble}
  \country{France}}
\email{vitalii.emelianov@inria.fr}

\author{Nicolas Gast}
\affiliation{
  \institution{Univ. Grenoble Alpes, Inria, CNRS, Grenoble INP, LIG}
  \city{Grenoble}
  \country{France}}
\email{nicolas.gast@inria.fr}

\author{Patrick Loiseau}
\affiliation{
  \institution{Univ. Grenoble Alpes, Inria, CNRS, Grenoble INP, LIG}
  \city{Grenoble}
  \country{France}}
\email{patrick.loiseau@inria.fr}

\begin{abstract}
  \input{abstract.tex}

\end{abstract}

  

  
\maketitle

\section{Introduction}
\label{sec:intro}
\input{./sec/intro.tex}

\section{Related work}
\label{sec:related work}
\input{./sec/related_work.tex}

\section{The model}
\label{sec:model}
\input{./sec/model.tex}

\section{Equilibrium characterization and resulting discrimination}
\label{sec:characterization}
\input{./sec/characterization.tex}

\section{Effects of the demographic parity mechanism on the selection}
\label{sec:fairness mechanism}
\input{./sec/affirmative_action.tex}

\section{Complementary results}
\label{sec:complements}

In the previous sections, we studied the properties of the Nash equilibrium of the game for large rewards $\s$ and showed that they can lead to discriminations. In this section, we complement this theoretical analysis by studying two (essentially independent) problems: First, does the population converge to the equilibrium? Second, what happens when the reward $\s$ is not large?

\subsection{Convergence to the Nash equilibrium}
\label{sec:convergence}
\input{./sec/convergence.tex}

\subsection{Selection problems with small and intermediate rewards}
\label{sec:small S}
\input{./sec/small.tex}

\section{Concluding discussion}
\label{sec:discussion}
\input{./sec/discussion.tex}

\begin{acks}
  This work has been partially supported by MIAI @ Grenoble Alpes (ANR-19-P3IA-0003), by the French National Research Agency (ANR) through grant ANR-20-CE23-0007 and grant ANR-19-CE23-0015.  
\end{acks}

\bibliographystyle{ACM-Reference-Format}
\bibliography{bibliography.bib}

% \newpage
\appendix

\section{Missing proofs}
\label{appendix:missing proofs}
\input{./sec/proofs.tex}

\end{document}

%% file: abstract.tex
%!TEX root = main.tex

To better understand discriminations and the effect of affirmative actions in selection problems (e.g., college admission or hiring), a recent line of research proposed a model based on \emph{differential variance}. This model assumes that the decision-maker has a noisy estimate of each candidate's quality and puts forward the difference in the noise variances between different demographic groups as a key factor to explain discrimination. The literature on differential variance, however, does not consider the strategic behavior of candidates who can react to the selection procedure to improve their outcome, which is well-known to happen in many domains. 

In this paper, we study how the strategic aspect affects fairness in selection problems. We propose to model selection problems with strategic candidates as a contest game: A population of rational candidates compete by \emph{choosing} an effort level to increase their quality. They incur a cost-of-effort but get a (random) quality whose expectation equals the chosen effort. A Bayesian decision-maker observes a noisy estimate of the quality of each candidate (with differential variance) and selects the fraction $\alpha$ of best candidates based on their posterior expected quality; each selected candidate receives a reward $S$. We characterize the (unique) equilibrium of this game in the different parameters' regimes, both when the decision-maker is unconstrained and when they are constrained to respect the fairness notion of demographic parity. Our results reveal important impacts of the strategic behavior on the discrimination observed at equilibrium and allow us to understand the effect of imposing demographic parity in this context. In particular, we find that, in many cases, the results contrast with the non-strategic setting. We also find that, when the cost-of-effort depends on the demographic group (which is reasonable in many cases), then it entirely governs the observed discrimination (i.e., the noise becomes a second-order effect that does not have any impact on discrimination). Finally we find that imposing demographic parity can sometimes increase the quality of the selection at equilibrium; which surprisingly contrasts with the optimality of the Bayesian decision-maker in the non-strategic case. Our results give a new perspective on fairness in selection problems, relevant in many domains where strategic behavior is a reality.

%% file: sec/intro.tex
%!TEX root = main.tex

% observe discriminations in selection problems
%\paragraph{\textbf{Discrimination in selection problems}}

Selection problems are problems in which, given multiple candidates, a decision-maker selects a fixed fraction of them with the objective of taking the best ones. Selection problems model a number of decision-making situations with high societal implications such as college admission or hiring. In such cases, it is important to consider the \emph{fairness} of the selection outcome---that is that there is no discrimination against demographic groups defined by \emph{sensitive} attributes. Yet, there is abundant evidence, including in recent years, of discrimination in both college admission and hiring, where the applicants are discriminated by salient demographic attributes such as gender \cite{hrw20,ahmed21}, race \cite{bertrand04,quillian17,gersen19}, or age \cite{cohen19}.
% \todo{add empirical evidence (incl old, and research papers).}

% explained by implicit bias or differential variance in old and recent literature (renewed interest)
%\paragraph{\textbf{Models of selection problems with implicit bias and differential variance}}

Recent literature on fairness in selection problem analyzed the problem using models based on two key ingredients to explain discrimination.  \citet{kleinberg18} model selection problems with \emph{implicit bias} (see also \cite{celis20,celis21,mehrotra22}), that is, where the decision-maker implicitly under-evaluates the quality of the candidates from disadvantaged demographic groups. On the other hand, \citet{emelianov20,emelianov22} and \citet{garg20}, following ideas from the economics literature on \emph{statistical discrimination} (see Section~\ref{sec:related work}), assume that the decision maker's estimate of the candidates quality is unbiased but has a higher variance for some demographic groups (a phenomenon terms \emph{implicit} or \emph{differential variance}). With both types of models, the authors study the discrimination that comes out of baseline decision-makers, and the impact of imposing fairness mechanisms such as the Rooney rule or the four-fifths rule.

% but the literature ignores strategic aspects. We know that people react
% in classification: no game
% Big question: how do strategic aspects affect discrimination
%\paragraph{\textbf{Selection problems with strategic candidates}}

In the above literature, the characteristics of the candidates used for selection (in particular, their qualities) are assumed to be fixed and exogenous---i.e., they do not depend on the selection procedure. In practice, however, candidates (i.e., individuals) involved in selection problems can \emph{adapt} to the selection procedure in order to increase the chances of a positive outcome. Such a \emph{strategic behavior} is observed in many domains \cite{woolley17,patro22}. A recent thread of literature on \emph{strategic classification} is devoted to modeling and analyzing the impact of this strategic behavior on classification problem \cite{Hardt16a,Dong18a,Milli19a,Kleinberg19a,Zhang19a,Miller20a,Tsirtsis20a,Braverman20a}. The selection problem, however, is fundamentally different from a classification problem in that the number of positive predictions is constrained---this will in particular lead to competition between individuals, see below. Moreover, this thread of literature did not investigate discrimination issues. This leaves open the key question, which is our focus in this work: 
%\begin{center}
\emph{How does the strategic behavior of candidates affect discrimination and the impact of fairness mechanisms  in selection problems?}
%\end{center}

\paragraph{\textbf{The selection problem with strategic candidates modeled as a contest game}}

We propose to model the selection problem with strategic candidates as a contest game (that is, roughly, as a game where candidates compete for a reward). We consider a population of candidates. Each candidate chooses an effort level $\m$ that they exert to improve their quality, at a quadratic cost $\cc\m^2/2$ (where $\cc$ is a constant coefficient). Each candidate has a latent quality $\w$ drawn randomly whose expected value is equal to their selected effort $\m$. A Bayesian decision-maker observes a noisy estimate $\wh$ of the quality $\w$ and selects a fraction $\ax$ of the best candidates based on the posterior expected quality $\wt=\E(\w|\wh)$. All selected candidates receive a reward of size $\s$, which is a quantitative measure of the benefit that the selection brings to a candidate (e.g., a job position or an education). 

In our model, we assume that the population of candidates is divided in two groups: the high-noise ($\A$) and the low-noise ($\B$) group (which refers to the noise in the estimate $\wh$ of the candidates quality). This group-dependent noise represents a form of information inequality common across different demographic groups: decision-makers are less familiar with candidates from certain groups (the $\A$-candidates), such that they are less able to precisely estimate the qualities of these candidates. This phenomenon is called \emph{differential variance} \cite{emelianov22, garg20}. It is at the basis of the economic theory of statistical discrimination and it was first described by \citet{phelps72} to explain the racial inequality of wages. In addition to group-dependent noise, we also consider a group-dependent cost coefficient (i.e., $\ca \neq \cb$). The group-dependent cost coefficient models socioeconomic inequalities; e.g., for students from low-income families it is typically harder to reach a desired level of quality due to costly preparatory courses.

\paragraph{\textbf{Overview of our results}}
Our model of a selection problem with strategic candidates defines a (population) game. We first show that this game has a unique Nash equilibrium. Then we focus primarily on characterizing the equilibrium in the regime of large rewards $\s$, which corresponds roughly to high-stake selection problems. In this regime, \textbf{we show that \emph{at equilibrium}, the following discrimination resulting from our model}:
\begin{itemize}
    \item If the cost coefficient is group-independent ($\ca=\cb$), then the high-noise candidates make a greater effort in average than the low-noise candidates. The latter are underrepresented\footnote{``Underrepresented'' in our paper means ``having less representation in the selection than its share in the candidates’ population''. This is a classical definition in the algorithmic fairness literature where the notion of demographic parity would mean that the two groups are equally represented.} in the selection, which is counterintuitive and is in contrast with the results in the non-strategic setting studied in \cite{garg20,emelianov22}.
    \item If the cost coefficient is group-dependent ($\ca\neq\cb$), then the noise level does not affect the outcome of the game. The cost-advantaged candidates (those for whom the cost coefficient is smaller) make a greater average effort compared to the cost-disadvantaged candidates, and the latter are underrepresented, irrespective of the noise. The noise is a second order effect compared to the cost difference, which totally dominates. This offers a potential explanation as to why low-noise candidates are often not underrepresented in practice: this is because the cost of effort for the high-noise candidates is usually large. 
\end{itemize}

In both cases stated above, one of the groups of candidates is always underrepresented. A potential remedy for this is the so-called \emph{demographic parity} mechanism, which imposes that the decision-maker selects candidates of all demographic groups at equal rates. Next, \textbf{we characterize the equilibrium when the Bayesian decision-maker is constrained by the demographic parity mechanism} (still in the regime of large rewards $\s$):
\begin{itemize} 
    \item We show that the demographic parity mechanism tends to equalize the average effort of the two groups compared to the unconstrained decision-maker: in most cases, the previously underrepresented group makes a greater average effort and the previously overrepresented group makes a lower average effort. 
    \item We characterize the change in the selection quality (or utility from the decision-maker perspective) from imposing demographic parity. Interestingly, we find that in some cases, the selection quality can improve compared to the unconstrained selection. This is surprising since, in the non-strategic setting, the unconstrained Bayesian decision-maker is optimal \cite{emelianov22}. Our result shows that it is no longer true in the strategic setting as in certain cases imposing a fairness mechanism can improve the selection quality, even against a Bayesian baseline. In other cases, we bound the degradation of quality that can result from imposing demographic parity.
\end{itemize} 

For the case of small rewards $\s$, \textbf{we get further analytical results on the equilibrium characterization}. We find that the results are different from that of the case of large $\s$, and are similar to the ones obtained in the non-strategic setting \cite{emelianov22,garg20}: the high-noise candidates are always underrepresented if the selection size $\ax$ is small enough. This indicates that, if the reward is small, then the impact of the strategic aspect (on discrimination results) is not important. We perform numerical experiments to illustrate the case of intermediate rewards $\s$ and how it matches the case of small and large $\s$ in their respective regimes.

Finally, \textbf{we study the convergence of two dynamics to the Nash equilibrium} (the best response and fictitious play dynamics). We observe that the trajectories of the best response dynamics converge to limit cycles that have a higher average effort (for both groups) than at equilibrium. The fictitious play dynamics seems to converge to equilibrium in our empirical results. 
% \todo{to be checked/adjusted} 
%; we study the properties of the cycles as functions of the reward $\s$.

\paragraph{\textbf{Implications of our results}}

Our results show that it is crucial to take into account the strategic nature of candidates involved in selection problem in order to understand discrimination and to predict the effect of imposing a fairness mechanism. They also show that discrimination in selection problems is a somewhat nuanced issue: it depends not only on the strategic aspect of the candidates but also on the range of assumed rewards and on the cost of effort of the different demographic groups. This means that a policy-maker, when considering whether to impose a fairness mechanism in a particular application, should evaluate the population of candidates (their costs of effort, etc.). It is worth noting, also, that we presented our results for a Bayesian decision-maker who computes a posterior estimate of the candidates quality (with knowledge of the group-dependent distributions). Our results are technically easy to extend to a group-oblivious decision-maker sorting the candidates by quality estimate $\wh$ irrespective of their group (another standard baseline); but in that case the high-noise and low-noise groups are reversed (we discuss that in Section~\ref{sec:discussion}). Hence, a policy-maker would also need to evaluate the baseline decision-maker they face.

%\paragraph{\textbf{Outline}}
\paragraph{Outline}

The rest of the paper is organized as follows. In Section~\ref{sec:related work}, we present the related literature. In Section~\ref{sec:model}, we formulate the problem in a game-theoretic framework. In Section~\ref{sec:characterization}, we show that, for the case of large rewards $\s$, the selection with strategic candidates leads to underrepresentation of one of the two groups of candidates.  In Section~\ref{sec:fairness mechanism}, we study how  demographic parity affects the incentives of candidates and the expected quality of the selection for large $\s$. In Section~\ref{sec:complements}, we complement our theoretical results by studying the convergence to equilibria and the case of small and intermediate rewards $\s$. We conclude with a discussion in Section~\ref{sec:discussion}.

%% file: sec/related_work.tex
%!TEX root = main.tex

%Due to space constraints, we describe here only briefly the most closely related works. %We provide additional details and additional references in Appendix~\ref{appendix:related work}.

\paragraph{\textbf{Statistical discrimination}}

The theory of statistical discrimination, initiated by \citet{phelps72} and \citet{arrow15}, considers the uncertainty of information about individual characteristics to explain racial/gender inequality in decision-making. \citet{phelps72} develops a model where each individual possesses a latent quality drawn from a fixed group-independent distribution. A Bayesian decision-maker observes a noisy estimate of individual's quality, where the noise is symmetric and zero-mean but has a group-dependent variance. The decision-maker assigns a wage equal to the expected posterior quality. This model is used to explain racial inequalities in wages. \citet{lundberg83} extend Phelps' model to a strategic setting by assuming two groups of workers that choose the values of effort according to a group-independent quadratic costs. The effort induces a quality that is assigned randomly but equal to the effort in expectation.  The authors show that, in the equilibrium, the high-noise candidates make lower effort and are payed less on average compared to the low-noise candidates; but that if the decision-maker is restricted to not use the group information for wage assignment, then the effort is equal for both groups.  In our work, we use a similar model but in the context of selection problems where a fraction of candidates receives a (fixed) reward rather than assigning a variable wage to all candidates as in \cite{phelps72,lundberg83}. We extend the model of \citet{lundberg83} to have group-dependent cost-of-effort, which as we shall see has a crucial effect on discrimination and we consider a different fairness mechanism  (demographic parity). 

Many statistical discrimination models (see e.g., \cite{arrow15,aigner77,coate93} and a survey in \cite{fang10}) assume that individuals of different groups have identical \emph{a priori} characteristics (e.g., cost-of-effort), but that the decision-maker uses their group-dependent belief when there is imperfect information to assess the performance of individuals in a group. In some cases, the discriminating beliefs (stereotypes) of decision-makers lead to equilibria in which these stereotypes are fulfilled. In contrast, we consider group-dependent cost-of-effort and group-dependent noise variance. We prove that the game in our model attains a unique equilibrium, and that discrimination occurs due to the group-dependent characteristics---i.e., if they become group-independent in our model, discrimination disappears.

\paragraph{\textbf{Selection problems in the non-strategic settings}}

Recent literature investigate statistical discrimination in selection problems, under the name of \emph{differential variance}. \citet{emelianov22} show that differential variance, in the case of Bayesian decision-maker, leads to the underrepresentation of the high-noise candidates (for selection fractions below $0.5$). They also study how quota-based fairness mechanisms (the 80\%-rule and the demographic parity) affect the fairness-utility tradeoff. \citet{garg20} study a similar setting but where the performance of candidates is measured by \emph{multiple} independent and unbiased estimates (from tests). The authors study how affirmative actions and access to testing affect the disparity level and the quality of the selected cohort. Some works also study the selection problem under implicit bias rather than differential variance---i.e., the decision-maker has a non-noisy but biased estimate of the candidates' qualities. \citet{kleinberg18} show that this type of bias naturally leads to underrepresentation of the disadvantaged group, and they study how a fairness mechanism called the Rooney rule affects the selection quality. This work is extended in \cite{celis20,celis21,mehrotra22}. 
Our work complements those studies by assuming \emph{strategic candidates} who can respond to a policy chosen by a Bayesian decision-maker. Similarly to \cite{garg20} and \cite{emelianov22}, we assume that the quality estimates are affected by differential variance, but we do not model bias \cite{kleinberg18} as we assume a Bayesian decision-maker who can correct for the bias. The main difference is that we assume that the quality distribution is not fixed but chosen by candidates to maximize their payoff which equals to the expected reward minus the cost-of-effort.

\paragraph{\textbf{Fairness in contests}}

%Contests are games in which players receive prizes according to their relative performance. 
A classical model of contest is given by \citet{lazear81}: individuals make a costly effort $\m$ that induces a noisy quality $\w=f(\m)+\varepsilon$, where $f$ is an increasing function of the effort; the player having the largest quality wins the prize $\s$. Fairness in contest games was studied from different perspectives in economics and computer science literature \cite{fu19}. \citet{schotter92} assume a two-player contest with a cost-advantaged and a cost-disadvantaged player---each pays a quadratic cost-of-effort but with different coefficients. They show that, at equilibrium, the cost-disadvantaged player makes a lower effort. Then they show that an affirmative action (adding a bonus to the score of the cost-disadvantaged player) leads all players to make lower efforts but increases the winning probability of the cost-disadvantaged player. %; lower efforts eventually lead to a lower payoff by the decision-maker in the equilibrium. 
Our model shares similarities with that of \cite{schotter92} (in particular group-dependent quadratic costs), but also has important differences: we consider an infinite population of candidates and we include group-dependent noise. Our results are also different as we analyze the dependence on $\s$ and we study another fairness mechanism (demographic parity)---in particular we find that it increases the effort of previously underrepresented group and sometimes even of both groups.

\paragraph{\textbf{Strategic ranking}}

Following the strategic classification literature (see previous section), \citet{liu21} study the \emph{strategic ranking} problem. They assume that each individual has a latent ranking and makes some costly effort to affect it. The effort induces a score $g$(effort)$\cdot$$f$(latent fixed rank), for some fixed strictly increasing functions $f$ and $g$, which in turn results in a post-effort ranking. \citet{liu21} study how different ranking reward functions (with and without randomization) of a fixed selection capacity $c$ affect the characteristics of individuals and of the selection at equilibrium (average welfare and scores). Finally, they consider a setting with two groups of candidates that differ in a multiplicative parameter $\gamma>1$ affecting the score and study the welfare gap for groups as a function of $c$.
In our model, we do not assume any pre-effort ranking---the score of an individual is purely determined by the effort and by a group-dependent noise. Our model is simpler compared to \cite{liu21} and is designed to capture the effects of the group-dependent cost-of-effort and noise in selection problems. This allows us to state concrete results involving the parameters of interest (the cost coefficient and noise variance). We also consider how a fairness mechanism (demographic parity) affects the group representations and the quality of selection at equilibrium.

%% file: sec/model.tex
% In this section, we define the strategies of candidates, as well as, the information available for candidates and the decision maker to perform a selection. In Section~\ref{ssec:players}, we define the main actors of the game, the candidates and the decision-maker. In Section~\ref{ssec:nash equilibrium}, we define the equilibrium of the game.

\subsection{Candidates and decision-maker}
\label{ssec:players}

\paragraph{\textbf{Candidates model}}
We assume a non-atomic game with a unit mass of candidates indexed by $\idx\in[0,1]$. There are two groups of candidates: $\A$ and $\B$. The letter $G\in\{\A,\B\}$ will denote any of these two groups, and the proportion of candidates from group $G$ is $\pg\in(0,1)$. Each candidate $\idx$ has a quality that depends on the \emph{effort} $\m_\idx$ that this candidate makes. In college admission or in hiring, the effort $\m_\idx$ of a candidate $i$ can be interpreted as candidate's achievements. It might, for instance, represent the number of courses followed by a student, the quality of number of degrees obtained, etc. We assume that a candidate $\idx$ that chooses to make an effort $\m_\idx\ge0$ has a quality $\w_\idx$ that is normally distributed with mean $\m_\idx$ and variance $\sq^2$: 
\begin{align*}
    \w_\idx \sim \norm(\m_\idx, \sq^2).
\end{align*}
Making an effort $\m$ costs a candidate from group $\g$ a quadratic cost $\cp\m^2/2$. The population-dependent cost coefficient $\cp$ can model socioeconomic factors like the income of the parents or the country of origin; these factors might make harder for some candidates than others to make the same effort $\m$.

\paragraph{\textbf{Decision-maker}}
A decision-maker wants to select the candidates having the largest qualities. To do so, the decision-maker observes a noisy estimate $\wh_\idx$ of the quality $\w_\idx$ of each candidate $\idx$:
\begin{align*}
    \wh_\idx = \w_\idx + \sx_{\g_\idx}\cdot\varepsilon_\idx,
\end{align*}
where $\varepsilon_\idx$ is a zero-mean centered random variable from the standard normal distribution $\norm(0,1)$; the noise variance $\sx_{\g_\idx}^2$ is assumed group-dependent.\footnote{Note that, for simplicity, we assume that the variance of the quality $\sqp^2$ does not depend on the candidate's group; Our results can be extended to the case of group-dependent $\sq_\g^2$ (see Section~\ref{sec:discussion}).} The quality estimate $\wh_\idx$ is a noisy but unbiased measurement of the quality $\w_\idx$ of a candidate $\idx$, e.g., an interview result. The group-dependent variance of noise $\sxp^2$ models the information inequality: if  interviewers are more familiar with candidates of some demographic groups, they are more confident in their evaluation compared to that of candidates of other groups.  Without loss of generality, we assume that $\sxa^2>\sxb^2$. We, thus, refer to $\A$-candidates as \emph{high-noise} candidates, and we refer to $\B$-candidates as \emph{low-noise} candidates. 

%The other factors  assumed random in our model (e.g., a school,  a family), can affect the quality. 

% where the expectation is calculated over the randomness of quality $\w$ given the value of a noise estimate $\wh$. 
% In the rest of the paper, we refer to $\ax$ as the \emph{selection rate}. 

We assume that the decision-maker knows\footnote{Our results can be extended to the case where the effort $\m$ is not observable (see Section~\ref{sec:discussion}).}  the effort $\m$ of all candidates, as well as the variances $\sqp^2$ and $\sxp^2$, and selects a proportion $\ax\in(0,1)$ of the candidates. The decision-maker aims at maximizing the expected quality of selected candidates and, therefore, selects the $\ax$ proportion having the largest expected quality $\wt$.  Using the property of conditional expectation for bivariate normal random variables, we can write the expectation of quality $\w_\idx$ given $\wh_\idx$ as
\begin{equation} 
    \label{eq:expected quality}
    \wt_\idx=\E(\w_\idx|\wh_\idx) = \wh_\idx \rho_{\g_\idx}^2 + (1 -\rho_{\g_\idx}^2) \m_\idx,
\end{equation}
where $\rho_{\g_\idx} = {\sqp}/{\sqrt{\sqp^2 + \sx_{\g_\idx}^2}} \in[0,1]$ is the correlation coefficient between the quality $\w_\idx$ and its estimate $\wh_\idx$. Since $\wt_\idx$ is a linear function of $\wh_\idx$, and $\wh_\idx$ is distributed normally, the expected quality $\wt_\idx$ follows a normal distribution with mean $\m_\idx$ and variance $\st_{\g_\idx}^2 = \sqp^4/\left(\sx_{\g_\idx}^2 + \sqp^2\right)$. Note that the larger the value of noise $\sx_{\g_\idx}^2$, the more the values of $\wt_\idx$ are concentrated around its mean value $\m_\idx$ (the smaller the variance $\st_{\g_\idx}^2$). From \eqref{eq:expected quality}, we observe that the decision-maker puts a higher weight on the effort $m$ for the high-noise candidates compared to the low-noise candidates: for the same level of effort, the high-noise candidates will be seen by the decision-maker as having less variability of expected quality compared to that of the low-noise candidates.

\subsection{The population game}
\label{ssec:nash equilibrium}

As the decision-maker selects the candidates having the largest expected quality, it will select all candidates whose expected quality  $\wt$ is larger than some selection threshold $\til$. We denote by $\pxg(m; \theta)$ the probability for a candidate of group $G$ to be selected when their effort is $\m$ and the selection threshold is $\til$. It equals:\footnote{In the rest of the paper, to simplify the notation, we will drop the index $\idx$.}
\begin{align}
    \label{eq:proba_selection}
    \px_\g\left(\m; \til\right) = \Pb\left(\wt\ge \til\right)=  \Phi
    \left(\frac{\m-\til}{\stp}\right),
\end{align}
where $\Phi$ is the cumulative distribution function of the standard normal distribution $\norm(0,1)$.

We assume that each selected candidate receives a positive reward $\s$, whereas the candidates who are not selected get a reward of $0$. Hence, the \emph{payoff} $\vp$ of a candidate from population $\g$ is
\begin{align}
    \label{eq: payoff}
    \vv_\g(\m;\til) & = \s \cdot \px_\g(\m;\til) - \cp\cdot\m^2/2.
\end{align}

%\paragraph{\textbf{Best response}}
Given a selection threshold $\til$, each candidate strategically decides on the effort $\m$ that maximizes their expected payoff $\vv_\g(\m;\til)$. Following the classical definition \cite{sandholm10}, we call it a \emph{pure best response} and we denote the set of all \emph{pure best response} strategies of a candidate as
\begin{align*}
    \br_\g(\til) = \left\{ m \text{ such that } \forall m'\ge0: \vp(\m; \til) \ge \vp(\m'; \til)\right\}.
\end{align*}
We say that the best response is unique if the set of best responses $\br_\g(\til)$ is reduced to a singleton. In such a case, by abuse of notation, we denote by $\br_\g(\til)$ the unique best response.

Similarly, by $\BR_\g(\til)$ we denote the set of all \emph{mixed best responses} of a candidate. It is the set of all probability distributions over the set of pure best responses $\br_\g(\til)$.

\paragraph{\textbf{Selection threshold}}

For each population $G\in\{\A,\B\}$, we denote by $\mu_G$ the distribution of efforts used by this population $G$. It is a probability distribution\footnote{If all candidates of the population $G$ make the same effort, $\mu_G$ is a \emph{pure strategy}. Otherwise, $\mu_G$ is a \emph{mixed strategy}.} on $\R^+$. We denote by $\mvec=(\mqa,\mqb)$ the effort distributions of the two populations. We denote the cumulative distribution function of the expected quality $\wt$ induced by $\mqg$ by $\F_{\mqg}$, and the cumulative distribution function of expected quality of the total population by $\F_{\mvec} = \pa \F_{\mqa} + \pb \F_{\mqb}$.  The decision-maker selects the $\alpha$-best candidates. Hence, it will select all candidate whose expected quality $\wt$ is above the $(1-\ax)$-quantile of the distribution $\F_{\mvec}$, that we denote by $\til(\mvec)= \F^{-1}_{\mvec}(1-\ax)$.
% The quantile $\til(\mvec)$ can be interpreted as a \emph{selection threshold}: the candidates are sorted by values of $\wt$, and the ones for which the expected value $\wt$ is greater than the the selection threshold $\til$ are selected.

%\paragraph{\textbf{The game}}

\paragraph*{\textbf{The game}}

The above definitions describe a population game between the candidates. We denote the game by $\gun$, where the superscript ``un'' emphasizes that the decision-maker is \emph{un}constrained,\footnote{In Section~\ref{sec:fairness mechanism}, we  study a  decision-maker faced with a fairness constraint (called demographic parity).} i.e., it selects the best $\ax$ candidates based on the expected quality $\wt$. Note that the game $\gun$ is parameterized by the reward size $\s$, the variance of quality $\sq^2$, the noise variances $\sxp^2$, the cost coefficients $\cp$, and the selection size $\ax$.

\subsection{Existence and uniqueness of the Nash equilibrium}

We use the standard definition of Nash equilibrium for populations games \cite{sandholm10}:
\begin{definition}[Nash equilibrium]
    A pair of effort distributions $\mvec = (\mqa,  \mqb)$ is called an equilibrium of the game $\gun$ if for all populations $G\in\{\A,\B\}$, the support of $\mq_\g$ is included in the set of best responses $\br_\g(\til(\mvec))$, where $\til(\mvec) = F^{-1}_{\mvec}(1-\ax)$.
\end{definition}

A Nash equilibrium is a situation where no candidate has an incentive to change its decision: if the population plays $\mvec$, then the selection threshold will be $\til(\mvec)$.  Hence, a candidate of a group $\g$ does not have an incentive to play a strategy that is not in $\br(\til(\mvec))$. As the support of $\mu_G$ is included in $\br(\til(\mvec))$, this implies that all candidates cannot obtain a higher payoff by unilaterally changing their decision.

In the rest of the paper, we will study the property of the Nash equilibrium of the game, which exists and is unique as guaranteed by the following theorem. 
\begin{theorem}
    \label{theorem:unique equilibrium}
    The game $\gun$ has a unique Nash equilibrium, that we denote by $\mvecun=(\mqaun, \mqbun)$.
\end{theorem}
We will also denote by $\tilun=F_{\mvecun}^{-1}(1-\ax)$ the selection threshold at equilibrium. Note that the equilibrium and the selection threshold depend on the game parameters, including the reward $\s$. In the following, we will study the properties of the equilibrium as $\s$ grows. For simplicity of exposition, we omit the dependency on $\s$ and  write $\mvecun$ or $\tilun$ instead of $\mvecun(\s)$ or $\tilun(\s)$.

We provide a detailed proof of Theorem~\ref{theorem:unique equilibrium} in Appendix~\ref{proof:unique equilibrium} whose main ingredient is to show that there is a one-to-one mapping between the equilibria of the game and the fixed points of a multi-valued function $\tiln$:
\begin{equation}
\label{eq:T}
    \tiln(\til)  = \{F^{-1}_{\mvec}(1-\ax): \mup\in \BR_\g(\til)\}.
\end{equation}
This simplifies the problem, as now we need to study the fixed points of a function of a single variable. By studying the first derivative of $\tiln$, we show that the function $\tiln$ has a unique fixed point and, hence, the defined game has also a unique equilibrium. 

\subsection{Summary of main notation}

We summarize the main notation in Table~\ref{table:notation}. Recall that, without loss of generality, we assume that $\sxa^2 > \sxb^2$ (which also implies that $\sta^2<\stb^2$). We, thus, refer to $\A$-candidates as \emph{high-noise}, and we refer to $\B$-candidates as \emph{low-noise}. To refer to the group with higher/lower cost-of-effort, we use the names \emph{cost-disadvantaged}/\emph{cost-advantaged}. If, for example, $\ca>\cb$, then we say that $\A$-candidates are \emph{cost-disadvantaged} and that $\B$-candidates are \emph{cost-advantaged}.
\begin{table}
    \caption{Summary of notation.}
    \begin{tabular}{ll}
        % symbol & meaning\\
        \toprule
        $\w_\idx\sim\norm(m_\idx, \eta^2)$ & quality of candidate $\idx$\\
        $\wh_\idx=\w_\idx + \varepsilon_\idx \sx_{\g_\idx}$ & quality estimate of candidate $\idx$  ($\varepsilon_\idx\sim\norm(0,1)$)\\
        % $\sxp^2$ & the variance of noise\\
        $\wt_\idx=\E(\w_\idx|\wh_\idx)$ & expected quality of candidate $\idx$\\
        % $\stp^2$ & the variance of of the expected quality\\
        $\mqg$ & distribution of effort for the population $\g$\\
        % $\wt_\g$ & distribution of expected quality  for the population $\g$ induced by $\mqg$\\
        $\mvec=(\mqa,\mqb)$ & distribution of effort for the total population\\
        $\F_{\mqg}$ & CDF of the expected quality $\wt$ for the population $\g$\\
        $\F_{\mvec}=\pa\F_{\mqa} + \pb\F_{\mqb}$ & CDF of the expected quality $\wt$ for the total population\\
        $\ax \in (0,1)$ & selection size\\ 
        $\til(\mvec) = \F_{\mvec}^{-1}(1-\ax)$ & selection threshold\\
        $\pg$ & proportion of $\g$-candidates in the total population\\
        $\pxg(m;\til)$ & selection probability for a $\g$-candidate  with effort $\m$ given the threshold $\til$\\
        \bottomrule
    \end{tabular}
    \label{table:notation}
\end{table}

%% file: sec/characterization.tex
In this section, we characterize the equilibrium of the game for large\footnote{In Section~\ref{sec:small S}, we show theoretical and numerical results for small and intermediate rewards $\s$.} rewards $\s$. We consider the case of large $\s$ as it models the competition in selection procedures with high stakes, e.g., hiring CEOs or college admission to high-ranked schools. 

\subsection{Properties of the best response}
\label{ssec:dropout}

We start with the characterization of the best response $\br_\g(\til)$ of a candidate. In Lemma~\ref{lemma:br large S} (whose proof is deferred to Appendix~\ref{proof:br large S}), we show that, when $\s$ is large, the best response $\br_\g(\til)$ is unique for all thresholds $\til$, except one, that we call a \emph{dropout threshold} $\tildg$.  We show that the value of the best response $\br_\g(\til)$ increases with the threshold $\til$ until the latter reaches $\tildg$; it then drops down and decreases when $\theta\ge\tildg$. This means that when the selection threshold $\til$ is too high, the candidates lose incentives to make any effort.  To emphasize the dependency of the dropout threshold  on the reward $\s$, we write $\tildg(\s)$.

\begin{lemma}[Best reponse for large $\s$]
  \label{lemma:br large S}
  There exists $\s_0$ such that for all rewards $\s\ge\s_0$, there exists a threshold $\tildg(\s)$, called the dropout threshold, such: 
  \begin{enumerate}[(i)]
    \item When the selection threshold is $\til = \tildg(\s)$, there are two pure best response strategies: $\br_\g(\tildg(\s)) = \left\{\mmin_\g\left(\tildg(\s)\right),\mmax_\g\left(\tildg(\s)\right)\right\}$. They satisfy:
      \begin{align*}
        \lims \mmin_\g\left(\tildg(\s)\right)= 0, \qquad \lims \frac{\mmax_\g\left(\tildg(\s)\right)}{\tildg(\s)} = 1.        
      \end{align*}
    \item For all $\theta\ne\tildg(\s)$, the pure best response $\br_\g(\theta)$ is unique. Moreover, for any $\gamma\in(0,1)$, we have 
    \begin{align*}
      \lims \br_\g\left(\tildg(\s) / \gamma\right)=0 \text{ and }\lims \frac{\br_\g\left(\gamma\tildg(\s)\right)}{\tildg(\s)}=\gamma.
    \end{align*}
  \end{enumerate}
\end{lemma}
Using the notation for asymptotic equivalence, we can write the statement of the theorem in a simpler form. For example, from the second part of (ii), we can infer that $\br_\g(\gamma \tildg) \sim \gamma \tildg $ for any $\gamma \in (0,1)$.
Therefore, the above lemma shows that for a given candidate, the best response $\br_\g(\til)$ increases essentially linearly up to the dropout threshold $\tildg$, and then drops to $0$ when the selection threshold is too high, i.e., $\til>\tildg$. We will illustrate this later in Fig.~\ref{fig:large S}.

Note that the dropout threshold $\tildg$ is group-dependent. As we will see later in Section~\ref{ssec:disrimination large s}, the most represented group at equilibrium will be the one having the largest dropout threshold. 
% Lemma~\ref{lemma:br large S} implies that candidates with a higher dropout threshold $\tildg$ are more likely to ``survive'' given an increasing competition (large $\til$). If the selection threshold $\til$ is too large (larger than the dropout $\tildg$), then candidates make nearly zero efforts at $\til$. Thus, it is important to study the dependency of the dropout threshold $\tildg$ on the cost coefficient $\cp$, and the noise variance $\sxp^2$ as these parameters are group-dependent.
Hence, in Lemma~\ref{lemma:dropout}, we show the asymptotic behavior of the dropout threshold as a function of $\s$. As we expect, the dropout threshold increases with the reward $\s$, and decreases with the cost coefficient $\cp$. A less intuitive property is its relation with the noise variance $\sxp^2$: we show that the dropout threshold increases as the noise variance $\sxp^2$ increases.  Lemma~\ref{lemma:dropout} is proven in Appendix~\ref{proof:dropout at infinity}.
\begin{lemma}[Dropout threshold for large $\s$]
\label{lemma:dropout}
  Let $\tildg(\s)$ be the dropout threshold.
  \begin{enumerate}[(i)] 
    \item For any $\cp$ and $\g$, the dropout threshold  can be expressed as 
    $$\tildg(\s) = \sqrt{2\s/\cp}\left(1 + o(1)\right) \text{ as }\s\to\infty.$$
    \item If $\ca=\cb$ and $\sxa^2 > \sxb^2$, then there exists $\s_0$ such that for all $\s\ge \s_0$,
  $\tilda(\s) > \tildb(\s).$
  \end{enumerate}
\end{lemma}

 From $(i)$, we conclude that the dependency of the dropout threshold $\tild$ on the cost coefficients $\cp$ is of higher order than the dependency on the noise variance  $\sxp^2$. In other words, when $\ca > \cb$, we have $\tilda(\s) < \tildb(\s)$ for large enough $\s$, regardless of the noise variance $\sxp^2$. When both groups have equal cost coefficients ($\ca = \cb$), the noise variance matters, and $\tilda(\s) > \tildb(\s)$ when $\sxa^2 > \sxb^2$ (or, equivalently, when $\sta^2<\stb^2$).

\subsection{Equilibrium strategy}
\label{ssec:main}

As we show in the previous section, the dropout threshold $\tildg$ is an important characteristics of candidates' best response: for thresholds $\til$ smaller than the dropout threshold $\tildg$, the candidates make an effort proportional to the threshold whereas for thresholds larger than the dropout, the candidates make nearly zero efforts.  In Theorem~\ref{theorem:equilibrium strategy} given below, we prove that, for large rewards, the selection threshold at equilibrium is equal to the dropout threshold of one of the two groups.  Note that this theorem is not stated in terms of the groups $\A$ and $\B$ but in terms of groups $\g_1$ and $\g_2$, where $\g_1$ is the group that has the largest dropout threshold, which holds if $C_{1}<C_2$ or ($C_1=C_2$ and $\sx_1^2>\sx_2^2$). Hence, the population $\g_1$ can correspond to the population $\B$ if $\ca > \cb$ (it corresponds, otherwise, to the population $\A$ as we assumed that $\sxa^2 > \sxb^2$). %To avoid cumbersome notations, we do not write explicitly the dependence on $\s$ of all quantities. 
\begin{theorem}[Equilibrium strategies]
  \label{theorem:equilibrium strategy}
  Fix $\ax \in (0,1)$ and two populations of candidates, $G_1$ and $G_2$, such that ($\cc_1 < \cc_2$) or ($\cc_1=\cc_2$ and $\sx_1^2 > \sx_2^2$). Then, there exists a reward $S_0$ such that for $\s\ge S_0$: 
  \begin{enumerate}[(i)]
    \item If $\ax \le \p_1$, then the equilibrium threshold $\tilun(\s)$ is $\tild_1(\s)$.  In this case:
    \begin{itemize}
      \item The $G_1$-candidates play a mixed strategy that consists in playing $\mmax_1$ with probability $\tau_1$ and $\mmin_1$, with probability $1-\tau_1$, where $\lims \tau_1 = {\ax}/{p_1}$.
      \item The $G_2$-candidates play the pure strategy $\mqbr_2(\tild_1(\s))$ and $\lims \mqbr_2(\tild_1(\s))= 0$.
    \end{itemize}
    \item If $\ax > p_1$, then the equilibrium threshold $\tilun(\s)$ is $\tild_2(\s)$. In this case:
    \begin{itemize}
      \item  The $G_1$-candidates play the pure strategy $\mqbr_1(\tild_2(\s))$ and $\lims \mqbr_1(\tild_2(\s)) / \tild_2(\s) = 1$. 
      \item The $G_2$-candidates play a mixed strategy that consists in playing $\mmax_2$ with probability $\tau_2$ and $\mmin_2$ with probability $1-\tau_2$, where $\lims \tau_2 = ({\ax - p_1})/{p_2}$.
    \end{itemize}
  \end{enumerate}
\end{theorem}

\begin{proof}[Proof Sketch]
  To simplify notation, we omit the dependence on $\s$ for all variables. Let us prove that the dropout threshold $\tild_1$ is the fixed point of $\tiln$ for the case $(i)$. We specify the efforts made in response to $\tild_1$, and we show that they lead to the same selection threshold $\tild_1$, i.e., $\tild_1$ is the fixed point of $\tiln$.
The case $(ii)$ can be proven similarly; we provide the complete proof in Appendix~\ref{proof:equilibrium strategy}. 

We can show that, as $\s\to\infty$, we have $\pxmax_\g\coloneqq\pxg(\mmax_\g(\tildg);\tildg)\xrightarrow{\s\to\infty} 1$ and $\pxmin_\g\coloneqq\pxg(\mmin_\g(\tildg);\tildg) \xrightarrow{\s\to\infty}0$ for all $\g\in\{\g_1,\g_2\}$. 
% In Lemma~\ref{lemma:dropout}, we also prove that for $\cc_1=\cc_2$ and $\sx_1^2>\sx_2^2$, we have $\tild_1 > \tild_2$ for substantially large $\s$. 

If $\ax \le p_1$, then, assume that $G_1$-candidates randomize their strategy by playing $\mmax_1$ with the probability $\tau_1$ and $\mmin_1$ with the probability $1-\tau_1$. The $G_2$-candidates play the deterministic strategy $\mqbr_2(\tild_1)\xrightarrow{\s\to\infty}0$. The probability $\tau_1$ can be found from the budget constraint: $\ax = p_1(\tau_1\pxmax_1 + (1-\tau_1)\pxmin_1) + p_2 \px_2(\tild_1,\mqbr_2(\tild_1)) \xrightarrow{\s\to\infty} p_1\tau_1$, so $\tau_1\xrightarrow{\s\to\infty}\ax/p_1$. The defined strategies satisfy the budget constraint, so $\tild_1$ is the fixed point of $\tiln$ and, hence, the defined distribution of effort is the equilibrium of  the  game $\gun$.
\end{proof} 

Fig.~\ref{fig:large S} illustrates the results of Theorem~\ref{theorem:equilibrium strategy}. In Fig.~\ref{fig:large S a}, we show the case of group-dependent noise but group-independent cost coefficient ($\ca=\cb$). In this case, the $\A$-candidates have a higher dropout threshold compared to the $\B$-candidates, hence, for $\tildb<\til\le\tilda$, the $\A$-candidates make a non-zero effort while the $\B$-candidates make a nearly-zero effort. In our illustration, the selection size $\ax=0.1$ is smaller than the proportion of $\A$-candidates $\pa=0.5$. We can verify that for $\til=\tilda$, if a proportion $\ax/\pa$ of $\A$-candidates plays $\mqamax$, and the rest of $\A$-candidates plays $\mqamin$, then such a strategy  satisfies the budget constraint.  Hence, $\tilda$ is the fixed point of the function $\tiln$, so $\tilun = \tilda$.

In Fig.~\ref{fig:large S b}, we illustrate the case of group-dependent cost coefficient ($\ca>\cb$). In this case, the $\B$-candidates have a higher dropout threshold compared to the $\A$-candidates, hence, for $\tilda<\til\le\tildb$, the $\B$-candidates make a non-zero effort while $\A$-candidates make a nearly-zero effort. For the purpose of illustration, we again assume that the selection size $\ax=0.1$ is smaller than the proportion of $\B$-candidates ($\pb=0.5$). Similarly to the previous case, we can verify that for $\til=\tildb$, if a proportion $\ax/\pb$ of $\B$-candidates plays $\mqbmax$, and the rest of $\B$-candidates plays $\mqbmin$, then this strategy satisfies the budget constraint, so $\tilun = \tildb$.
% code to generate figures: /code/Fig\ 1.ipynb
\begin{figure}
  \centering
  \begin{subfigure}{.5\textwidth}
    \centering
    \begin{tikzpicture}[scale=1.33] % scale = figure width / 0.6
      \node (img)  {\includegraphics[width=.8\linewidth]{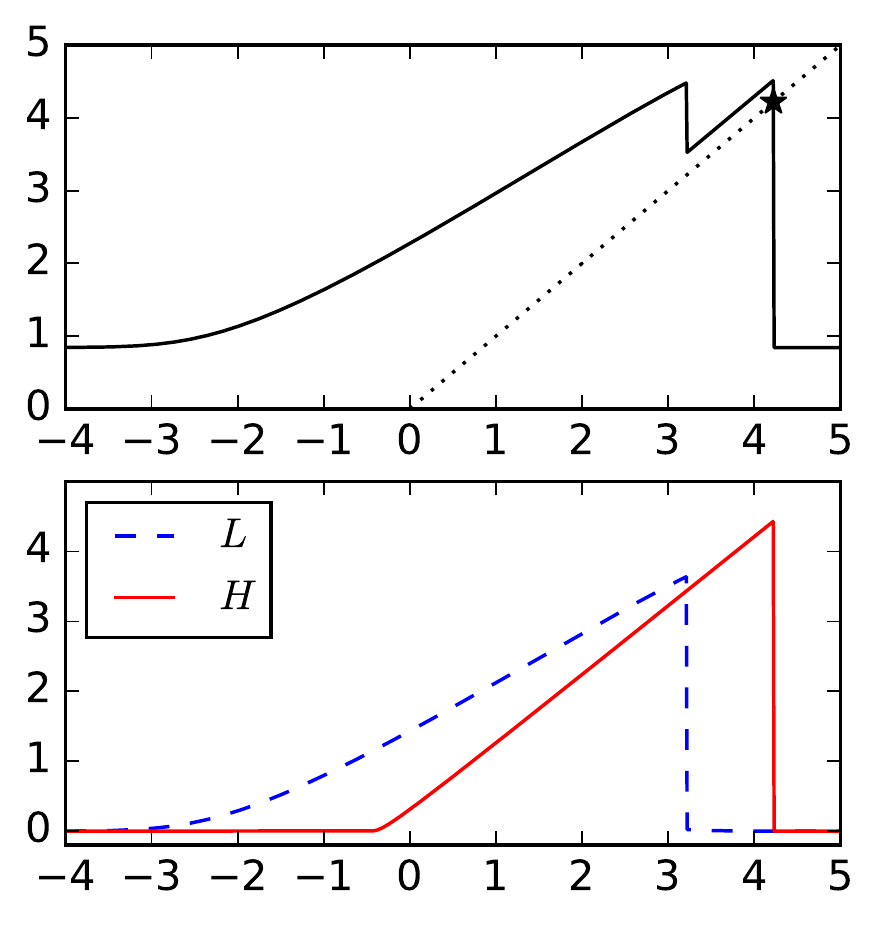}};
      \node[below=of img, node distance=0cm, yshift=3.5em,font=\color{black}] {\tiny$\theta$};
      % \node (ne) at (0.6, 0.2)[rotate=60] {\tiny $T(\til)=\til$};
      \node[left=of img, node distance=0cm, rotate=90, anchor=center,yshift=-3em, xshift=-3.5em,font=\color{black}] {\tiny$\br$};
      \node[left=of img, node distance=0cm, rotate=90, anchor=center,yshift=-3em,xshift=4em, font=\color{black}] {\tiny$\tiln$};
      \node (maxb) at (1.1, -0.4) [font=\color{blue}]{\tiny $\mmax_{\B}$};
      \node (minb) at (1.0, -1.6) [font=\color{blue}]{\tiny $\mmin_{\B}$};
      \node (maxa) at (1.72, -0.24) [font=\color{red}]{\tiny $\mmax_{\A}$};
      \node (mina) at (1.45, -1.6) [font=\color{red}]{\tiny $\mmin_{\A}$};
      \node (ta) at (1.7, -2.1) [font=\color{red}]{\tiny $\tilda$};
      \node (tb) at (1.2, -2.1) [font=\color{blue}]{\tiny $\tildb$};
      \node (tun) at (1.75, 1.55) [font=\color{black}]{\tiny $\tilun$};
      \node (id) at (1.0, 1.0) [font=\color{black}, rotate=40]{\tiny $\mathrm{id}(\theta)$};
     \end{tikzpicture}
    \caption{$\ca/\cb=1$}
    \label{fig:large S a}
  \end{subfigure}%
  \begin{subfigure}{.5\textwidth}
    \centering
    \begin{tikzpicture}[scale=1.33] % scale = figure width / 0.6
      \node (img)  {\includegraphics[width=.8\linewidth]{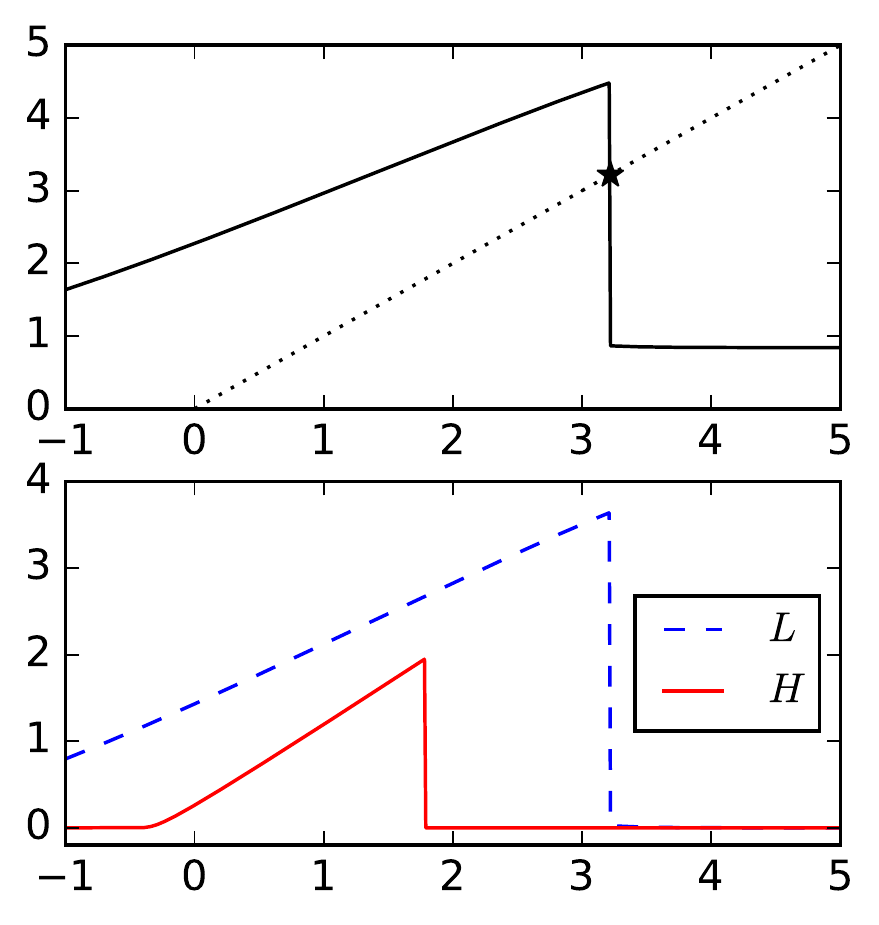}};
      \node[below=of img, node distance=0cm, yshift=3.5em,font=\color{black}] {\tiny$\theta$};
      % \node (ne) at (0.6, 0.2) [rotate=60]{\tiny $T(\til)=\til$};
      \node[left=of img, node distance=0cm, rotate=90, anchor=center,yshift=-3em, xshift=-3.5em,font=\color{black}] {\tiny$\br$};
      \node[left=of img, node distance=0cm, rotate=90, anchor=center,yshift=-3em,xshift=4em, font=\color{black}] {\tiny$\tiln$};
      % \node (ne) at (1.4, 1.0) {\tiny $\mmqbun$};
      \node (maxb) at (1.0, -0.3) [font=\color{blue}]{\tiny $\mmax_{\B}$};
      \node (minb) at (1.0, -1.5) [font=\color{blue}]{\tiny $\mmin_{\B}$};
      \node (maxa) at (0.2, -0.9) [font=\color{red}]{\tiny $\mmax_{\A}$};
      \node (mina) at (0.2, -1.5) [font=\color{red}]{\tiny $\mmin_{\A}$};
      \node (ta) at (-0.1, -2) [font=\color{red}]{\tiny $\tilda$};
      \node (tb) at (0.85, -2) [font=\color{blue}]{\tiny $\tildb$};
      \node (tun) at (1.0, 1.3) [font=\color{black}]{\tiny $\tilun$};
      \node (id) at (0.2, 0.8) [font=\color{black}, rotate=30]{\tiny $\mathrm{id}(\theta)$};
     \end{tikzpicture}
    \caption{$\ca / \cb = 5$}
    \label{fig:large S b}
  \end{subfigure}
  \caption{Best response functions and the Nash equilibrium for $\s=10$, $\cb=1$, $\sta=0.1$ and $\stb=1$, $\pa=0.5$ and $\ax=0.1$. Both figures illustrate the case (i) of Theorem~\ref{theorem:equilibrium strategy}. The dotted line on the uppermost panels is the identity function
  $\mathrm{id}(\theta)=\theta.$}
  \label{fig:large S}
\end{figure}
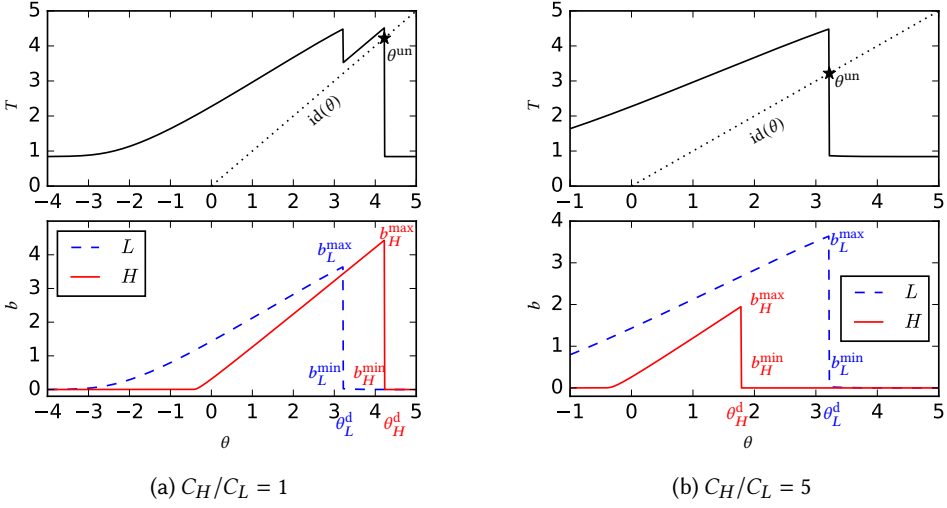

\subsection{Discrimination due to the group-dependent variance and cost}
\label{ssec:disrimination large s}

In Theorem~\ref{theorem:equilibrium strategy}, we show that the equilibrium distribution of effort $\mvecun$ depends on the relation between the dropout thresholds $\tildg$ for different populations, $\A$ and $\B$. We now study how the noise variance $\sxp^2$ and the cost coefficient $\cp$ affect the representation of groups in the selection at equilibrium. Overall, we show that the group-dependent noise variance and the group-dependent cost coefficient can both lead to underrepresentation of some group of candidates at equilibrium. 
% The relation between the dropout threshold depends on the relation between the costs $\cp$ and the variance $\stp^2$ which is given in Lemma~\ref{lemma:dropout}. 

For a given population $\g$, we denote by $\mmqgun$ the \emph{average effort} that candidates from group $\g$ exert at equilibrium, and we denote by $\mpxgun$  the \emph{selection rate} which is the probability for a randomly chosen candidate from group $\g$ to be selected. Since $\px_\g(m; \theta)$ is the probability for a $\g$-candidate to be selected when exerting an effort $m$ (see \eqref{eq:proba_selection}), these quantities satisfy:
\begin{align*}
  \mmqgun = \int_{m\ge0} m\, d\mqgun(\m); \qquad\qquad
  \mpxgun = \int_{m\ge0}\px_\g(\m; \tilun)\,d\mqgun(\m).
\end{align*}
We say that the group $\A$ (or $\B$) is \emph{underrepresented} if $\mpx_{\A} < \mpx_{\B}$ (or $\mpx_{\B} < \mpx_{\A}$). Note that when we say ``underrepresented', it means that a demographic group has less representation in the selection than its share in the population of all candidates. This definition is not conditioned on the assumption that the mean of the qualities $\w$ is the same for both groups, but this assumption does make the notion of demographic parity more obviously appealing since there is no a priori distinction in the average quality between the groups. Nevertheless, we do not claim that demographic parity would be justified only under this assumption.

In Theorem~\ref{theorem:large S} below, we show that if the cost coefficient $\cp$ is group-independent ($\ca=\cb$), then the high-noise $\A$-candidates make a larger effort compared to that of the low-noise $\B$-candidates, and the latter are underrepresented. However, if the cost coefficient is group-dependent ($\ca\ne\cb$), the noise variance does not play a role if the reward $\s$ is large enough: the  cost-advantaged candidates make larger effort compared to that of the cost-disadvantaged candidates, and as a result, the latter are underrepresented. This theorem also shows that, as $\s$ goes to infinity, the ratios of efforts and selection rates can grow unbounded. For example, in the case $(i)$ of Theorem~\ref{theorem:large S}, if there are enough $\A$-candidates to fill the selection budget $\ax$, then $\B$-candidates will asymptotically not be selected when $\s$ goes to infinity. We will see in Section~\ref{sec:small S}, that for moderate values of reward $\s$, the $\B$-candidates still have some representation in the selection, but that can be very small. 

\begin{theorem}[Discrimination for large rewards $\s$]
  \label{theorem:large S}
  % Fix $\ax\in(0,1)$, $\ca$, $\cb$ $\sxa^2$ and $\sxb^2$. 
  Let $\mqaun$ and $\mqbun$, $\pxaun$ and $\pxbun$ be the equilibrium effort distributions and selection rates of the game $\gun$.
  \begin{enumerate}[(i)]
  \item If $\ca < \cb$ or if $\ca=\cb$ and $\sxa^2>\sxb^2$, then there exists $\s_0$ such that for all $\s\ge \s_0$ the high-noise $\A$-candidates make greater effort compared to the low-noise $\B$-candidates, and the latter are underrepresented:
 \begin{align*}
  \lim_{\s\to\infty} \frac {\mmqbun}  {\mmqaun} = \lim_{\s\to\infty}\frac{\mpxbun}{\mpxaun}=
  \begin{cases}
     0 & \text{if } \pa \ge \ax,\\
     \frac{\ax - \pa}{1-\pa} < 1 & \text{if } \pa < \ax.
  \end{cases}
 \end{align*}
  \item If $\ca > \cb$, then there exists $\s_0$ such that for all $\s\ge \s_0$ the cost-advantaged $\B$-candidates make a greater effort compared to the cost-disadvantaged $\A$-candidates, and the latter are underrepresented: 
  \begin{align*}
    \lim_{\s\to\infty}\frac {\mmqaun} {\mmqbun}= \lim_{\s\to\infty} \frac {\mpxaun}  {\mpxbun}=
    \begin{cases}
       0 & \text{if } \pb \ge \ax,\\
       \frac{\ax - \pb}{1-\pb}  < 1 & \text{if } \pb < \ax.
    \end{cases}
   \end{align*}
  \end{enumerate}
\end{theorem}
\begin{proof}
To prove $(i)$, we compute the average effort and the selection rate given the equilibrium mixed strategies defined in Theorem~\ref{theorem:equilibrium strategy}. We consider the case of $\ax\le\pa$ and the case of $\ax > \pa$ separately. 
% We denote $\cc=\ca=\cb$.

If $\ax \le \pa$, then 
$\mmqaun = \mqamax\tau_\A + \mqamin(1-\tau_\A) = \ax/\pa\cdot\sqrt{2\s/\ca}(1+o(1))$, $\s\to\infty$ and  $\mmqbun = \br_\B(\tilda) \xrightarrow{\s\to\infty}0$, so $\mmqbun / \mmqaun \xrightarrow{\s\to\infty}0$.  Similarly, the expected selection rates for $\A$- and $\B$-candidates: 
\begin{align*}
    \mpxaun &= \pxamax \tau_\A + \pxamin(1-\tau_\A) \xrightarrow{\s\to\infty} \ax/\pa,\\
    \mpxbun &= \pxb(\tilda, \mqbbr(\tilda)) \xrightarrow{\s\to\infty}0,
\end{align*}
which implies that $\mpxbun / \mpxaun\xrightarrow{\s\to\infty}0$.

If $\ax > \pa$, then we, again, consider the strategies in the equilibrium:
\begin{align*}
    \mmqaun &= \br_\A(\tildb) = \sqrt{2\s/\cb}(1+o(1)),\;\s\to\infty,\\ 
    \mmqbun &= \mqbmax \tau_\B + \mqbmin (1-\tau_\B) = (\ax - \pa)/\pb\cdot \sqrt{2\s/\cb}(1+o(1)),\;\s\to\infty.
\end{align*}
Similarly, for the selection rates, we have: $\mpxaun = \pxa(\tildb, \mqabr(\tildb)) \xrightarrow{\s\to\infty} 1$, and 
$\mpxbun = \pxbmax\tau_\B + \pxbmin(1-\tau_\B) \xrightarrow{\s\to\infty}(\ax-\pa)/\pb$. 
The proof of $(ii)$ is identical to the proof of $(i)$ so we omit it.
\end{proof}

%% file: sec/affirmative_action.tex
In Theorem~\ref{theorem:large S}, we show that the equilibrium in the game $\gun$ leads to underrepresentation of one of the groups of candidates: for the group-independent cost coefficient ($\ca=\cb$), the low-noise candidates are underrepresented; for the group-dependent cost coefficient ($\ca\ne\cb$), the cost-disadvantaged candidates are underrepresented. 

To reduce the inequality of representation, colleges (and employers) sometimes perform \emph{affirmative actions}. This can take the form of quotas for low-income or minority groups or any forms of promotions. Among those affirmative actions are ones that make sure that the selection rates for the groups are close,  meaning that the proportions of both groups in the selection should be close to proportions of the groups in the total population. The strongest condition among those is the \emph{demographic parity} (see \cite{barocas19}) which requires that selection rates for both populations must be equal: 
%\begin{align*}
    $\mpx_{\A}=\mpx_{\B}$.
%\end{align*}
This implies, given the budget constraint $\mpx_{\A}\pa + \mpx_{\B}\pb=\ax$, that $\mpx_{\A}=\mpx_{\B}=\ax$.

\subsection{The game with the demographic parity mechanism}

The demographic parity mechanism removes the competition among the groups of candidates, $\A$ and $\B$, as from each group $\g$ a proportion $\ax\in(0,1)$ must be accepted.  Hence, the game with the demographic parity mechanism, that we denote by $\gdp$, can be represented as two independent games, $\gun_\A$ and $\gun_\B$, each defined for a single group, $\A$ or $\B$. Note that for this game, the selection thresholds will be different for the two groups: for a group $\g$, the decision-maker will select all candidates whose expected quality is greater or equal than $\til(\mup)=\F^{-1}_{\mup}(1-\ax)$, which corresponds to the $\ax$ best fraction of this group---it is the quantile of the distribution of expected qualities induced by $\mup$ and not the one induced by $\mvec$ as in the unconstrained case.

The equilibrium of the game $\gdp$ where the Bayesian decision-maker has a demographic parity constraint is defined as follows.
\begin{definition}[Nash equilibrium of the game $\gdp$]
    A pair of effort distributions $\mvec = (\mqa,  \mqb)$ is an equilibrium of the game with the demographic parity constraint $\gdp$ if for both groups $\g\in\{\A,\B\}$:
    \begin{align*}
        \mup \in \BR_\g(\til(\mup)),\text{ where $\til(\mup)=\F^{-1}_{\mup}(1-\ax)$.}
    \end{align*}
\end{definition}
Mimicking the unconstrained case, we denote by $\mvec^\mathrm{dp}=(\mqadp, \mqbdp)$ the equilibrium of this game, and by $\tilgdp=F_{\mqpdp}^{-1}(1-\ax)$ the group-dependent equilibrium selection threshold. The superscript ``dp'' emphasizes that the decision-maker is \emph{demographic parity}-constrained.

Since the game $\gdp$ can be represented as two separate games, $\gun_\A$ and $\gun_\B$, that has a unique Nash equilibrium according to Theorem~\ref{theorem:unique equilibrium}, the game $\gdp$ also has a unique Nash equilibrium.

\subsection{Efforts induced by the demographic parity mechanism}

In the previous section, we showed that, for large rewards $S$, it is possible that only one of the two groups makes a positive effort while the other group considers that the game is not worth playing because of a too unfair competition.  The situation is radically different with demographic parity mechanism as each candidate competes with similar candidates. As we show below, the demographic parity mechanism pushes the previously underrepresented group to make more effort than before. Moreover, it can also push the previously overrepresented group to make more effort.

In the first theorem below,  we characterize the equilibrium strategy of the game with the demographic parity mechanism  $\gdp$. This result is a direct corollary of Theorem~\ref{theorem:equilibrium strategy} as we consider two separated games with an unconstrained decision-maker.
\begin{theorem}[Equilibrium strategy with the demographic parity mechanism]
    \label{theorem:dp strategy}
    There exists $\s_0$ such that for all $\s \ge \s_0$, the equilibrium of the game with the demographic parity mechanism  $\gdp$ is a pair of distributions $\mvec^\mathrm{dp} =(\mqadp, \mqbdp)$ where for each group $\g\in\{\A,\B\}$, $\mqpdp$ consists in playing $\mmax_\g$ with probability $\tau_\g$ and $\mmin_\g$ with probability $1-\tau_\g$, where $\lims \tau_\g = \ax$.
\end{theorem}

\begin{proof}
    According to the demographic parity mechanism, the selection rate per each group $\g$ must be equal to $\mpx_\g=\ax$.
    Hence, this theorem can be seen as a special case of Theorem~\ref{theorem:equilibrium strategy} but with a single group of mass $1$ out of which we need to select the best $\ax\in(0,1)$. By applying directly the result of Theorem~\ref{theorem:equilibrium strategy}, we show that the proposed strategy is the equilibrium strategy.
\end{proof}

In Corollary~\ref{corollary:dp efforts} given below, we compare the efforts made by two groups at equilibrium. We show that, if the cost coefficient is group-independent, then the demographic parity mechanism equalizes the effort as $\s$ grows (together with the selection rates as $\mpxadp=\mpxbdp$ by definition). For group-dependent cost coefficient, the cost-disadvantaged $\A$-candidates make lower average effort compared to that of $\B$-candidates  yet the average effort ratio is bounded by $\sqrt{\ca/\cb}$. This is in contrast to Section~\ref{ssec:disrimination large s} where we show that the average effort ratio can be unbounded in the case of the unconstrained decision-maker.
\begin{corollary}[Equilibrium effort ratio in $\gdp$]
    \label{corollary:dp efforts}
    Let $\mqadp$ and $\mqbdp$ be the equilibrium effort distributions in the game with the demographic parity mechanism $\gdp$. The average efforts of both populations $\g\in\{\A,\B\}$ satisfy:
    \begin{align*}
        \lims \frac{\mmqadp} {\mmqbdp} = \sqrt{\frac{\cb}{\ca}}.
    \end{align*}
    In particular, if $\ca=\cb$, then the average efforts of both populations grow at equal rates. If $\ca \ne \cb$, then the cost-disadvantaged candidates make a lower average effort compared to that of the cost-advantaged candidates. 
\end{corollary}

\begin{proof}
Using the equilibrium strategies of the game  with the demographic parity mechanism $\gdp$ from Theorem~\ref{theorem:dp strategy}, we show that:
$
    \lim_{\s\to\infty} \frac{\mmqadp}{\mmqbdp} = \lim_{\s\to\infty} \frac{\tau_\A \mqamax +(1-\tau_\A)\mqamin}{\tau_\B \mqbmax +(1-\tau_\B)\mqbmin} = \sqrt{\frac{\cb}{\ca}}.
$
\end{proof}

% \todo{proof for $\mqaun \ge \mqadp$}

By reducing the competition between groups, affirmative action policies are often criticized because they might encourage individuals to make less effort, which reduces the overall quality of the selected candidates. We show below that, in fact, the demographic parity mechanism always encourages the previously underrepresented group to make a larger effort than in the unconstrained case. For the previously overrepresented group, the situation varies: in many situations, the candidates from the overrepresented group will make a lower effort than in the unconstrained case, but when $\sqrt{\cb/\ca} < \ax$ and $\ax>\pb$, we show that they make a higher effort compared to that of the unconstrained case.  Note that the last result may seem rather counterintuitive as demographic parity reduces the competition between groups.  Later, in Section~\ref{ssec:quality_dp}, we  show the implications of this result on the average quality of the selected candidates. 
\begin{corollary}[Equilibrium effort ratio in $\gun$ vs. $\gdp$]
    \label{theorem:game comparison}
    Let $\mqaun$ and $\mqbun$ be the equilibrium effort distributions in the unconstrained game $\gun$. Similarly, let $\mqadp$ and $\mqbdp$ be the equilibrium effort distributions in the game with the demographic parity constraint $\gdp$. 
    \begin{enumerate}[(i)]
        \item If $\ca=\cb$ and $\sxa^2>\sxb^2$, then
        \begin{align*}
            \lims \frac{\mmqaun}{\mmqadp} = \begin{cases}
                1/\pa & \text{if }\ax \le \pa,\\
                1/\ax & \text{if }\ax > \pa,
            \end{cases}
            \text{ and }
            \lims \frac{\mmqbun}{\mmqbdp} = \begin{cases}
                0 & \text{if }\ax \le \pa,\\
                \frac{\ax-\pa}{\ax-\ax\pa}& \text{if }\ax > \pa.
            \end{cases}
        \end{align*}
        \item If $\ca>\cb$, then
        \begin{align*}
            \lims \frac{\mmqaun}{\mmqadp} = \begin{cases}
                0 & \text{if }\ax \le \pb,\\
                \frac{\ax-\pb}{\ax-\ax\pb} & \text{if }\ax > \pb,
            \end{cases}
            \text{ and }
            \lims \frac{\mmqbun}{\mmqbdp} = \begin{cases}
                1/\pb & \text{if }\ax \le \pb,\\
                \sqrt{\frac{\cb}{\ca}}\frac 1 \ax& \text{if }\ax > \pb.
            \end{cases}
        \end{align*}
    \end{enumerate}
\end{corollary}
\begin{proof}[Proof Sketch]
    Similarly, as in the proof of Corollary~\ref{corollary:dp efforts}, we calculate the limits  by using the equilibrium strategies found in Theorem~\ref{theorem:equilibrium strategy} and in Theorem~\ref{theorem:dp strategy}. The full proof is in Appendix~\ref{proof:game comparison}.
\end{proof}

\subsection{Selection quality with and without the demographic parity mechanism}
\label{ssec:quality_dp}

We now show the implication of the previous result on how the demographic parity mechanism affects the selection quality. In a non-strategic setting (see e.g., \cite{emelianov22}), the unconstrained decision maker is optimal in expectation. We show here that this no longer holds in the strategic setting: there exist scenarios under which \emph{the Bayesian decision-maker is not optimal} in terms of the expected quality of selection. In such settings, the demographic parity constraint leads to a more qualified cohort. This counterintuitive phenomenon is due to the fact that demographic parity induces less competition between groups but a more fair competition within each group (compared to the unconstrained case). 

% In Corollary~\ref{corollary:dp efforts}, we show that the demographic parity mechanism induces the disadvantaged population to make an effort comparable to that of the advantaged population. Moreover, for group-independent cost coefficient, both populations make almost equal effort for large $\s$.

% We define the expected quality of selection as follows.
% \begin{definition}[Expected selection quality]
% The expected quality of the selection by the unconstrained decision-maker is
% \begin{align}
%     \label{eq:revenue un}
%     \rev = \sum_{G\in\{A,B\}} \pg\cdot \E\left(\wg \cdot[\wtg \ge \til(\mvecun)]\right),
% \end{align}
% where $\til(\mvecun) = F^{-1}_{\mvecun}(1-\ax)$.
% \todo{$\wg$ and $\wtg$ are not defined}

% The expected quality of the selection by the demographic parity-constrained decision-maker is
% \begin{align}
%     \label{eq:revenue dp}
%     \revdp = \sum_{G\in\{A,B\}} \pg\cdot \E\left(\wg \cdot[\wtg \ge \til(\mqpdp)]\right),
% \end{align}
% where  $\til(\mqpdp) = F^{-1}_{\mqpdp}(1-\ax)$. 
% \end{definition}

We denote by $\rev$ and by $\revdp$ the expected quality at equilibrium of the selected candidates for the unconstrained and the demographic parity constrained games.  In the theorem below, we characterize the ratio of the equilibrium cohort qualities by the unconstrained decision-maker from Section~\ref{sec:characterization}, and the demographic parity constrained decision-maker that we study in this section.  For group-independent cost coefficient ($\ca = \cb$), we show that the quality ratio $\rev/\revdp$ tends to $1$ in the limit of large $\s$. For group-dependent cost coefficient ($\ca\neq \cb$), we show that the quality ratio can be smaller than one---the Bayesian decision-maker is not optimal if the candidates are strategic, and the demographic parity mechanism can lead to a better-qualified cohort. 
% This result is in contrast with the selection outcome in the non-strategic setting where the Bayesian decision-maker is optimal \cite{emelianov22}.
\begin{theorem}[Selection quality ratio for $\gun$ and $\gdp$]
    \label{theorem:revenue ratio}
    Let $\rev$ and $\revdp$ be the expected qualities of selection at equilibrium of the game $\gun$ and of the game $\gdp$, respectively.
    \begin{enumerate}[(i)]
        \item If $\ca=\cb$, then the ratio of the expected quality given by the unconstrained and the demographic parity constrained decision-makers tends to one with $\s$:
        $
            \lims \frac{\rev}{\revdp} = 1.
        $
        \item If $\ca >  \cb$, then the ratio of the expected quality given by the unconstrained and the demographic parity constrained decision-makers can be smaller than one. Formally, for $c=\sqrt{\cb/\ca}$:
        \[
            \lims \frac{\rev}{\revdp}  = 
            \begin{cases}
            \frac{1}{c\pa + \pb} > 1 & \text{if } \pb \ge \ax,\\
            \frac{c}{c\pa + \pb} < 1 & \text{if } \pb < \ax.
            \end{cases}
        \]
    \end{enumerate}
\end{theorem}

\begin{proof}[Proof Sketch]
    First, we prove that as $\s\to\infty$, the expected quality of the selected cohort grows at equal rate with the expected effort: $\E(\wg\cdot[\wtg \ge \til]) \sim \bar\mq_\g$. 
    Hence, using the equilibrium strategies from  Theorem~\ref{theorem:equilibrium strategy} and  Theorem~\ref{theorem:dp strategy}, we can estimate the ratio $\rev/\revdp$ in the limit of $\s\to\infty$. The complete proof is given in Appendix~\ref{proof: quality ratio}.
\end{proof}

We emphasize that the condition under which the demographic parity mechanism improves the average selection equality is when the selection rate $\alpha$ is larger than the size of the cost-advantaged group. The improvement of selection quality due to the demographic parity mechanism is explained by the fact that, without the demographic parity constraint, the advantaged minority has no incentives to make a large effort because the competition includes a lot of cost-disadvantaged candidates. Once the competition is among candidates of each separate populations, the cost-advantaged candidates have to compete with other cost-advantaged candidates, so they have to make a larger effort to be selected. 

The demographic parity can decrease the average selection quality when $\ax\le\pb$ and when both groups have different cost coefficients $\cc_\g$. In this case, if the low-noise $\B$-candidates  are the majority, then the ratio of quality $\rev/\revdp$ cannot be larger than $2$ as $\s$ goes to infinity, regardless of the cost coefficients.

%% file: sec/convergence.tex
To answer the first question, we perform a series of numerical experiments in which the decisions are made repeatedly. At a given time $t$, the candidates consider past data to make a strategic decision. This could represent, for instance, the case of college admission where candidates consider the distribution of grades from previous years; in this example, each decision epoch is a different year. 

We study two population dynamics: \emph{best response} and \emph{fictitious play}.
\begin{itemize}
  \item For the \emph{best response} dynamics, at each of the discrete times $t\in\{1,2,\dots,T\}$, the candidates observe the strategy played at the previous time step $\mvec^{(t-1)}$ and play a best response to it:
  \begin{align*}
    \m_\g^{(t)} \in \br_\g(\til(\mvec^{(t-1)})).
  \end{align*}
  \item For the \emph{fictitious play} dynamics, at each of the discrete times $t=1,2,\dots,T$, the candidates observe the whole history of plays and assume that the distribution of efforts is the empirical distribution of effort from time $1$ to $T$. Candidates then play a best response to it:
  \begin{align*}
    \m_\g^{(t)} \in \br_\g(\til(\hat\mvec^{(t)})),
  \end{align*}
  where $\til(\hat\mvec^{(t)}) = \F^{-1}_{\hat\mvec^{(t)}}(1-\ax)$ and $\hat\mvec^{(t)}=\sum_{s=1}^{t-1}\frac{1}{t-1} \mvec^{(s)}$.  
\end{itemize}
%where $\til(\mvec^{(t-1)}) = \F^{-1}_{\mvec^{(t-1)}}(1-\ax)$.

We numerically evaluate these two policies and report the results in Fig.~\ref{fig:cycles}.  For the best response dynamics, we observe that $\bm{\m}^{(t)}=(\m_\A^{(t)}, \m_\B^{(t)})$ converges to a limit cycle for any starting point.  This is because when $S$ is large,\footnote{For $\s < \frac{1}{2}\cp\stp^2/\phi(1)$ we can show that the function $\tiln$ is a contraction mapping, so any trajectory of the best response dynamics converges to the Nash equilibrium.} the best response map is not continuous (recall Fig.~\ref{fig:large S}). The period of the limit cycle increases with the reward size $\s$ but the behavior is similar for all $\s$: starting from $(0,0)$, the candidates from both populations increase the effort as time increases. Then, the competition becomes too high and one of the populations drops out, i.e., make almost zero effort. After, the competition is only among the candidates of a single population until it becomes too difficult and all candidates drop out and return to the initial state. The cycle ends here, and the new cycle starts.  In Fig.~\ref{fig:br 10} and \ref{fig:br 100}, we also plot the average trajectory $\bar{\bm{\m}}^{(t)} = (\bar\m_\A^{(t)}, \bar\m_\B^{(t)})$, where $\bar\m_\g^{(t)}=\frac{1}{t-1}\sum_{s=1}^{t-1}\m_\g^{(s)}$. We observe that the average over the trajectory seems to converge, yet the average effort over the trajectory is significantly larger than that of the average equilibrium effort for both groups. 
% code to generate: /code/Fig\ 3.ipynb
\begin{figure}
  \centering
  \begin{subfigure}{.24\textwidth}
    \centering
    \begin{tikzpicture}[scale=0.22/.3]
      \node (img)  {\includegraphics[width=0.9\linewidth]{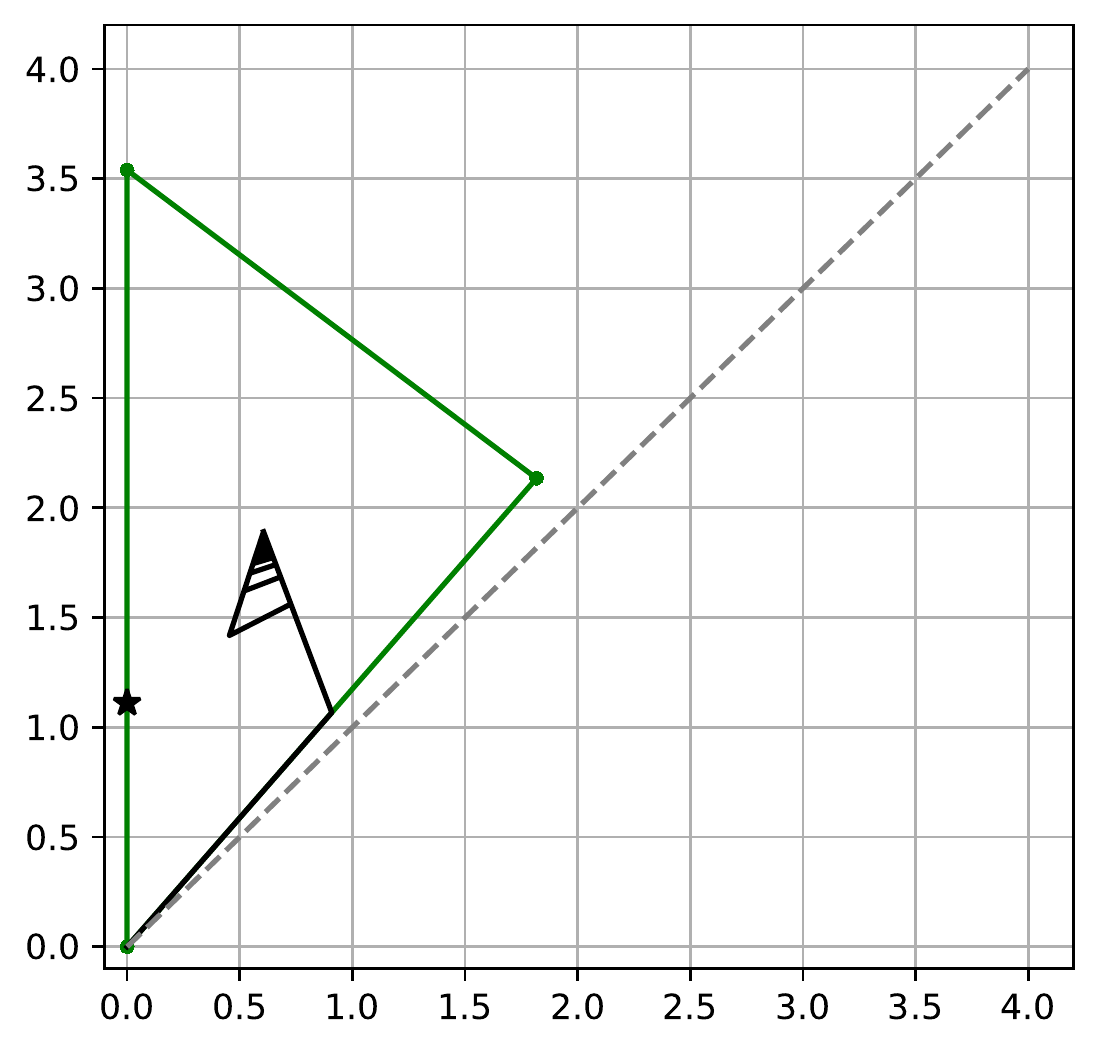}};
      \node[below=of img, node distance=0cm, yshift=3.5em,font=\color{black}] {\tiny$\m_{\A}$};
      \node[left=of img, node distance=0cm, rotate=90, anchor=center,yshift=-3em,font=\color{black}] {\tiny$\m_{\B}$};
      \node (mcycl) at (-0.6, 0.2) {\tiny $\bar\m^{(t)}$};
      \node (mt) at (-0.1, 0.8) {\tiny $\m^{(t)}$};
      \node (ne) at (-1.1, -0.6) {\tiny $\mmq^\mathrm{un}$};
      % \node (diag) at (1.0, 1.1)[rotate=45,font=\color{black}] {\tiny $\mmq_{\A}=\mmq_{\B}$};
     \end{tikzpicture}
    \caption{BR, $\s=10$}
    \label{fig:br 10}
  \end{subfigure}
  \begin{subfigure}{.24\textwidth}
    \centering
    \begin{tikzpicture}[scale=0.22/.3]
      \node (img)  {\includegraphics[width=0.9\linewidth]{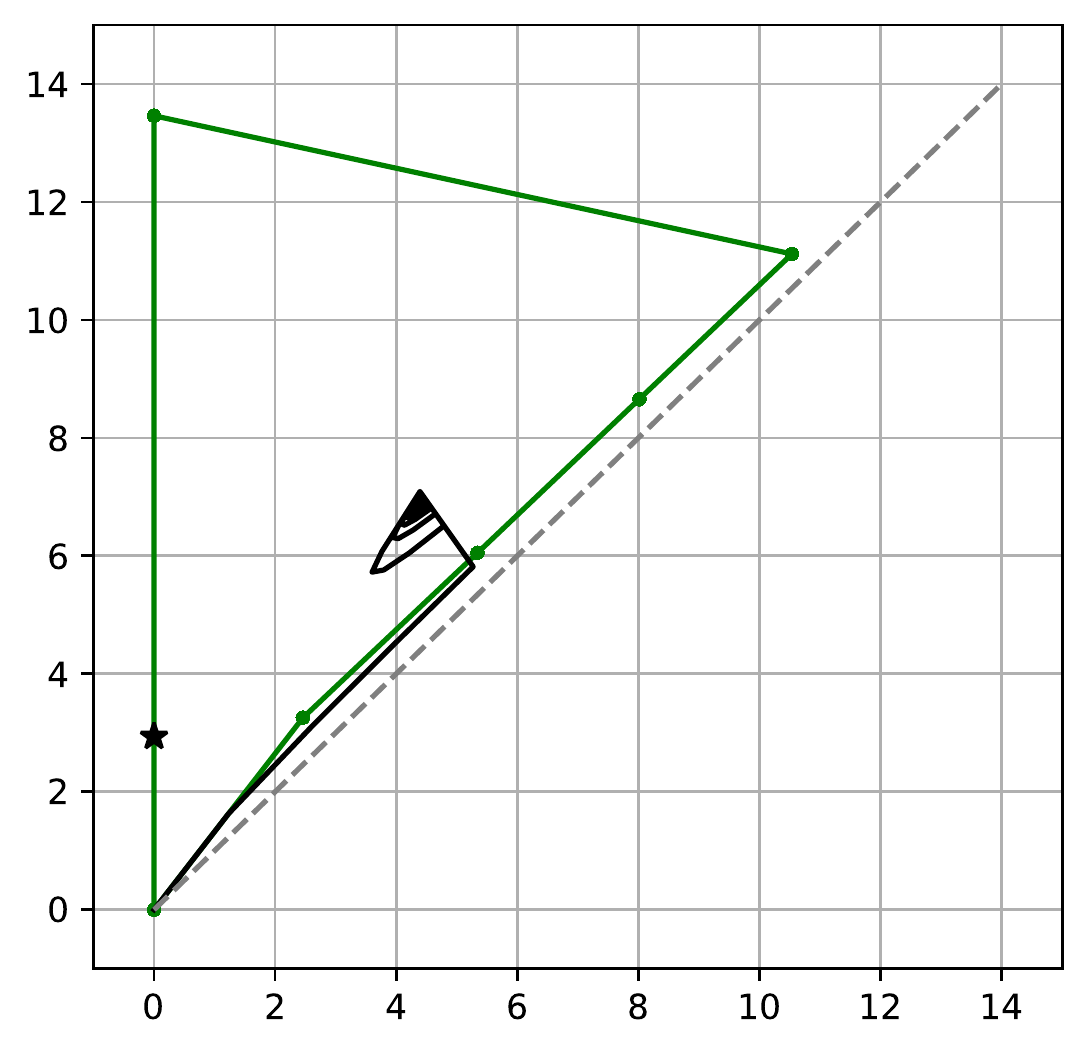}};
      \node[below=of img, node distance=0cm, yshift=3.5em,font=\color{black}] {\tiny$\m_{\A}$};
      \node[left=of img, node distance=0cm, rotate=90, anchor=center,yshift=-3em,font=\color{black}] {\tiny$\m_{\B}$};
      \node (mcycl) at (-0.4, 0.3) {\tiny $\bar\m^{(t)}$};
      \node (mt) at (-0.1, 1.4) {\tiny $\m^{(t)}$};
      \node (ne) at (-1.1, -0.6) {\tiny $\mmq^\mathrm{un}$};
      % \node (diag) at (1.3, 1.35)[rotate=45,font=\color{black}] {\tiny $\mmq_{\A}=\mmq_{\B}$};
     \end{tikzpicture}
    \caption{BR, $\s=100$}
    \label{fig:br 100}
  \end{subfigure}
  % \begin{subfigure}{.3\textwidth}
  %   \centering
  %   \begin{tikzpicture}
  %     \node (img)  {\includegraphics[width=0.9\linewidth]{../fig/cycle-S-1000.pdf}};
  %     \node[below=of img, node distance=0cm, yshift=3.5em,font=\color{black}] {\tiny$\m_{\A}$};
  %     \node[left=of img, node distance=0cm, rotate=90, anchor=center,yshift=-3em,font=\color{black}] {\tiny$\m_{\B}$};
  %     \node (mcycl) at (-0.4, 0.2) {\tiny $\bar\m^{(t)}$};
  %     \node (mt) at (-0.1, 1.2) {\tiny $\m^{(t)}$};
  %     \node (ne) at (-1.1, -0.6) {\tiny $\mmq^\mathrm{un}$};
  %     % \node (diag) at (1.2, 1.3)[rotate=45,font=\color{black}] {\tiny $\mmq_{\A}=\mmq_{\B}$};
  %    \end{tikzpicture}
  %   \caption{$\s=1000$}
  %   \label{fig:characterization c}
  % \end{subfigure}
  \begin{subfigure}{.24\textwidth}
    \centering
    \begin{tikzpicture}[scale=0.22/.3]
      \node (img)  {\includegraphics[width=0.9\linewidth]{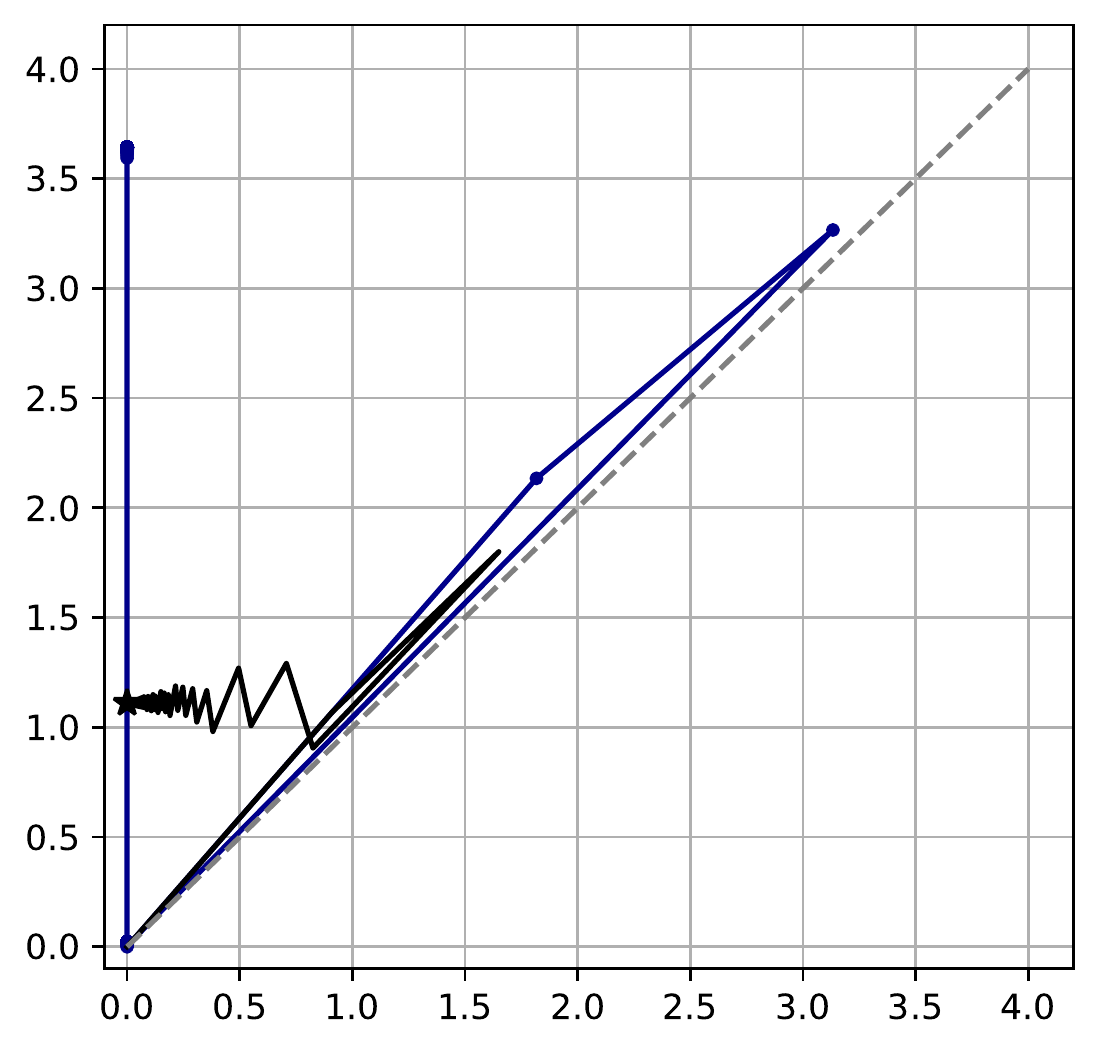}};
      \node[below=of img, node distance=0cm, yshift=3.5em,font=\color{black}] {\tiny$\m_{\A}$};
      \node[left=of img, node distance=0cm, rotate=90, anchor=center,yshift=-3em,font=\color{black}] {\tiny$\m_{\B}$};
      \node (mcycl) at (-0.8, -0.2) {\tiny $\bar\m^{(t)}$};
      \node (mt) at (0.3, 1) {\tiny $\m^{(t)}$};
      \node (ne) at (-1.2, -0.9) {\tiny $\bar\mq^\mathrm{un}$};
      % \node (diag) at (1.0, 1.1)[rotate=45,font=\color{black}] {\tiny $\mmq_{\A}=\mmq_{\B}$};
     \end{tikzpicture}
    \caption{FP, $\s=10$}
    \label{fig:fp 10}
  \end{subfigure}
  \begin{subfigure}{.24\textwidth}
    \centering
    \begin{tikzpicture}[scale=0.22/.3]
      \node (img)  {\includegraphics[width=0.9\linewidth]{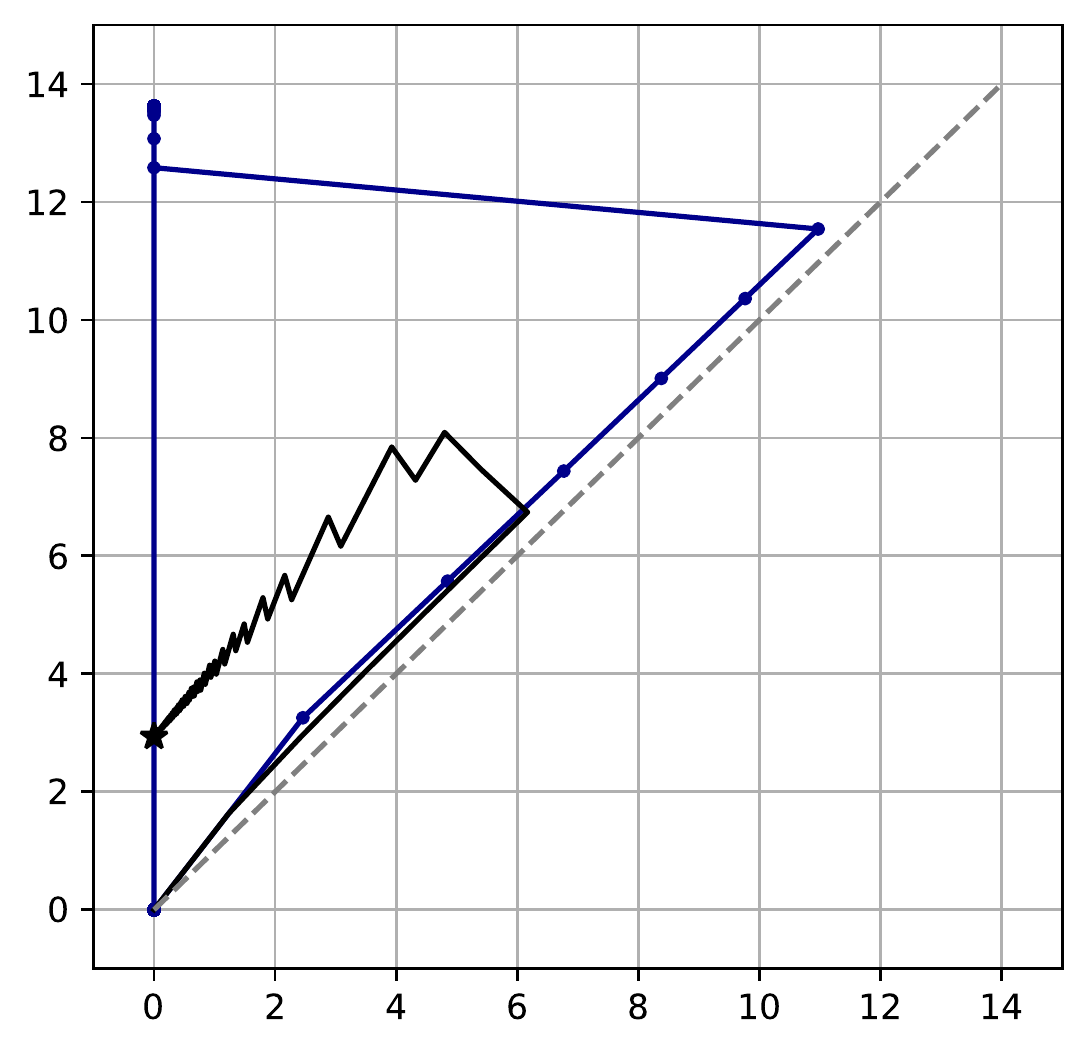}};
      \node[below=of img, node distance=0cm, yshift=3.5em,font=\color{black}] {\tiny$\m_{\A}$};
      \node[left=of img, node distance=0cm, rotate=90, anchor=center,yshift=-3em,font=\color{black}] {\tiny$\m_{\B}$};
      \node (mcycl) at (-0.4, 0.5) {\tiny $\bar\m^{(t)}$};
      \node (mt) at (-0.2, 1.4) {\tiny $\m^{(t)}$};
      \node (ne) at (-1, -0.7) {\tiny $\mmq^\mathrm{un}$};
      % \node (diag) at (1.3, 1.35)[rotate=45,font=\color{black}] {\tiny $\mmq_{\A}=\mmq_{\B}$};
     \end{tikzpicture}
    \caption{FP, $\s=100$}
    \label{fig:fp 100}
  \end{subfigure}
  \caption{Best response (BR) and fictitious play (FP) dynamics for different rewards $\s$. The parameters of simulations are $T=500$, $\ca=1.5$, $\cb=1$, $\pa=\pb=0.5$, $\ax=0.1$, $\sta=0.6$ and $\stb=1$.}
  \label{fig:cycles}
  \end{figure}

The case of fictitious play dynamics is different, and it is depicted in Fig.~\ref{fig:fp 10} and \ref{fig:fp 100}. In this case, the empirical distribution of efforts \emph{does} converge to the Nash equilibrium.  This is illustrated on the figure by the fact that the empirical average of effort converges to the average value of effort of the Nash equilibrium: $\bar{\m}_\g^{(t)} \xrightarrow{t\to\infty}\mmqgun$ for both groups $\g\in\{\A,\B\}$. Note that this is not a pointwise convergence but rather a convergence to a cycle: at equilibrium, the strategy played by the $\B$-candidates converges to a cycle on the values $b^{\min}_\B$ and $b^{\max}_\B$.

%% file: sec/small.tex
%  $\s<\cp\stp^2/\phi(1)$  %% problem with latex if inside title of subsection.

We start with the case of small rewards $\s<\cp\stp^2/\phi(1)$ for which the results are quite different from the ones obtained for large $\s$ in Section~\ref{sec:characterization}. We show that the cost ratio $\cb/\ca$, as well as the expected quality variance ratio $\stb^2/\sta^2$, play significant  roles in the outcome of the game. If the cost for the $\A$-candidates is too high compared to that of the $\B$-candidates, then the $\A$-candidates make lower effort for all selection sizes $\ax$. Otherwise, if the cost for $\A$-candidates is comparable to that of $\B$-candidates (e.g., $\ca=\cb$) or it is lower, then for both small and large selection sizes $\ax$, the high-noise $\A$-candidates make less effort compared to that of low-noise $\B$-candidates. For both cases, there exists a parameter dependent value of $\ax_0$ such that for all values of $\ax\le \ax_0$, the high-noise $\A$-candidates  are underrepresented.
Interestingly, in (ii) of Theorem~\ref{theorem:disparity of effort}, we also observe that it is possible that $\A$-candidates make a larger effort than $\B$-candidates, yet they are selected at a lower rate, $\mpxaun < \mpxbun$. 

Overall, the high-level interpretation of the below result is that \emph{for small enough rewards $\s$ and for small enough values of $\ax$, the high-noise $\A$-candidates are always underrepresented}. This result is in contrast with the results for large $\s$ studied in Section~\ref{sec:characterization}, and it is similar to the result in the non-strategic setting studied in \cite{emelianov22,garg20}.
\begin{theorem}[Discrimination for small rewards $\s$]
\label{theorem:disparity of effort}
    Assume that $\s < {\cp\stp^2} / {\phi(1)}$ for all $\g\in\{\A,\B\}$.  Let $\mqaun$ and $\mqbun$, $\pxaun$ and $\pxbun$ be the equilibrium efforts and selection rates of the game $\gun$. Denote $K_\mq \coloneqq \sqrt{\frac{-2\ln (\ca\sta/(\cb\stb))}{1/\sta^2 - 1/\stb^2}}$ and $K_\px\coloneqq\sqrt{\mathcal{W}\left(\frac{\s^2\left(\frac{1}{\ca\sta} - \frac{1}{\cb\stb}\right)^2}{2\pi(\stb-\sta)^2}\right)}$, where $\mathcal{W}$ is the Lambert function defined as the inverse to the function $f(\lambda) = \lambda e^\lambda$ and $\Phi^c$ is the complementary cumulative distribution function of the standard normal distribution $\norm(0,1)$.
    \begin{enumerate}[(i)]
    \item If $\ca\sta > \cb\stb$, then $\mmqaun < \mmqbun \text{ for all }\ax\in(0,1)$, and $\mpxaun < \mpxbun$ if and only if $\ax < \Phi^c\left(-K_x\right)$.
    \item If $\ca\sta \le \cb\stb$, then 
    \begin{align*}
    &\mmqaun < \mmqbun  \iff \ax \in \left(0,  \sum_{G\in\{A,B\}} \pg \Phi^c\left( K_\mq/\stp\right)\right)  \cup \left( \sum_{G\in\{A,B\}} \pg \Phi^c\left({-K_\mq}/\stp\right), 1\right),\\
    &\mpxaun < \mpxbun \iff \ax < \Phi^c\left(K_x\right). %< 1/2.
    \end{align*}
    \end{enumerate}

\end{theorem}

\begin{proof}[Proof Sketch]
    We show that for $\s<\cp\stp^2/\phi(1)$ the best response in pure strategies is unique, hence the equilbirium effort distribution is a singleton $\mqgun = \delta(\m-\m_\g^\mathrm{un})$. Note that the first-order condition on a maximum of the payoff function $\vv_\g$ is also a sufficient condition; it can be written: $$\frac{\s}{\stp}\phi\left(\frac{\m^\mathrm{un}_\g - \tilun}{\stp}\right) - \cc_\g{\m^\mathrm{un}_\g}=0 \iff \m^\mathrm{un}_\g = \frac{\s}{\cc_\g\stp}\phi\left(\frac{\m^\mathrm{un}_\g - \tilun}{\stp}\right) .$$
    Since we aim to find a value of $\ax$ when $\m_{\A}^\mathrm{un}  = \m_{\B}^\mathrm{un}$, we equate the right-hand sides of the above equation for two groups, and, by solving this equation, we obtain the values of $\tilun -\m_{\g}^\mathrm{un}$. By substituting this expression to the budget constraint, we derive the value of $\ax$ at  which $\m^\mathrm{un}_\A=\m^\mathrm{un}_\B$:
    \begin{align*}
      \ax = \sum_{\g\in\{\A,\B\}} \pg \Phi^c\left(\frac{\tilun -\m_{\g}^\mathrm{un}}{\stp}\right).
    \end{align*}
    The proof for $\pxgun$ is similar to that of $\m_\g^\mathrm{un}$. A complete proof is given in Appendix~\ref{proof:inequality of effort}. 
\end{proof}
% Overall, we observe that higher variance of estimates is detrimental to the effort for competitive selection: given equal costs $\ca=\cb$, the $\A$-candidates have less incentives than $\B$-candidates to compete if the selection size $\ax$ is small. However, if the cost ratio $\ca/\cb$ is large enough (larger than $\stb/\sta$ > 1), then the $\A$-candidates make lower effort than the $\B$-candidates for any $\ax\in(0,1)$.

% Also note that given the assumption that ${\s} < {\cp\stp^2}/ {\phi(1)}$, the transition value of $\ax$ specified in the theorem above has the following bounds for all $\s$, $\ca=\cb=\cc$:
% \begin{align*}
%     0.5 \ge \Phi^c\left(\sqrt{W\left(\frac{\s^2}{2\pi \cc^2\sta^2\stb^2}\right)}\right) > \Phi^c(1)\approx 0.159.
% \end{align*}

% \subsection{Filling the gap between small and large rewards via numerical simulations}
% \label{ssec:numerical}

% Finally, to match the case of small $\s<\cp\stp^2/\phi(1)$ and large $\s$, we perform numerical experiments for intermediate values of $\s$.  

\paragraph{Intermediate rewards}

To conclude our analysis, we fill the gap between our theoretical results for the cases of small and large rewards $\s$ using numerical simulations.\footnote{The code can be found at \url{https://gitlab.inria.fr/vemelian/strategic-selection-code}.} 
We perform our numerical simulations for the values of reward $\s=1,10,100,1000$. The simulation result for $\s=1$ is studied theoretically in the first part of this section, as $\s=1$ satisfies the condition $\s < \cp\stp^2/\phi(1)$. The result for $\s=\infty$, studied in Section~\ref{sec:characterization} and Section~\ref{sec:fairness mechanism}, is represented in Fig.~\ref{fig:characterization} using a black solid line.

We plot the ratio of $\mpxbun/\mpxaun$ for the case of group-independent cost coefficient in Fig.~\ref{fig:characterization a}, the ratios of $\mpxaun/\mpxbun$ and $\rev/\revdp$ for the case of group-dependent cost coefficient in Fig.~\ref{fig:characterization b} and Fig.~\ref{fig:characterization c}. Overall, we observe a relatively smooth transition between the two regimes of $\s$ in all figures. In addition, the behavior for $\s=100$ and for $\s=1000$ is quite close to the behavior for $\s=\infty$.
\begin{figure}
  \centering
  \begin{subfigure}{.3\textwidth}
    \centering
    \begin{tikzpicture}
      \node (img)  {\includegraphics[width=.9\linewidth]{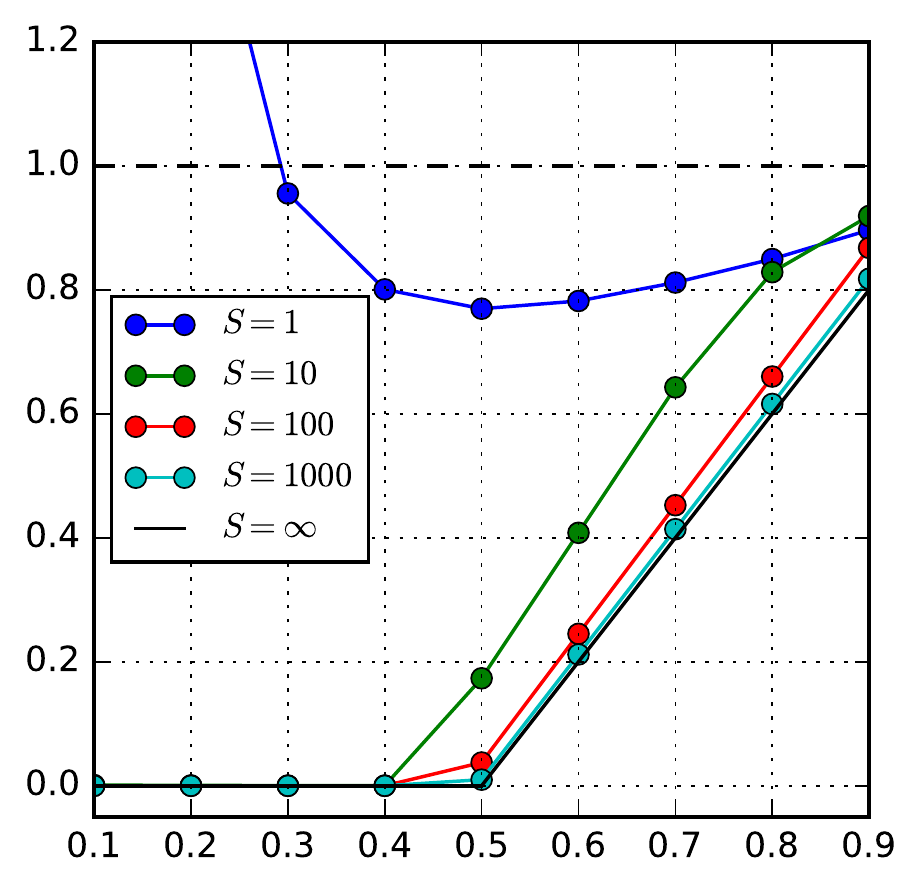}};
      \node[below=of img, node distance=0cm, yshift=3.5em,font=\color{black}] {\tiny$\ax$};
      \node[left=of img, node distance=0cm, rotate=90, anchor=center,yshift=-3em,font=\color{black}] {\tiny$\mpxbun/\mpxaun$};
      \node (ne) at (1.3, -0.6)[font=\color{black}] {\tiny $\frac{\ax-\pa}{\pb}$};
     \end{tikzpicture}
    \caption{$\mpxbun / \mpxaun$ for $\ca/\cb=1$}
    \label{fig:characterization a}
  \end{subfigure}
  \begin{subfigure}{.3\textwidth}
    \centering
    \begin{tikzpicture}
      \node (img)  {\includegraphics[width=.9\linewidth]{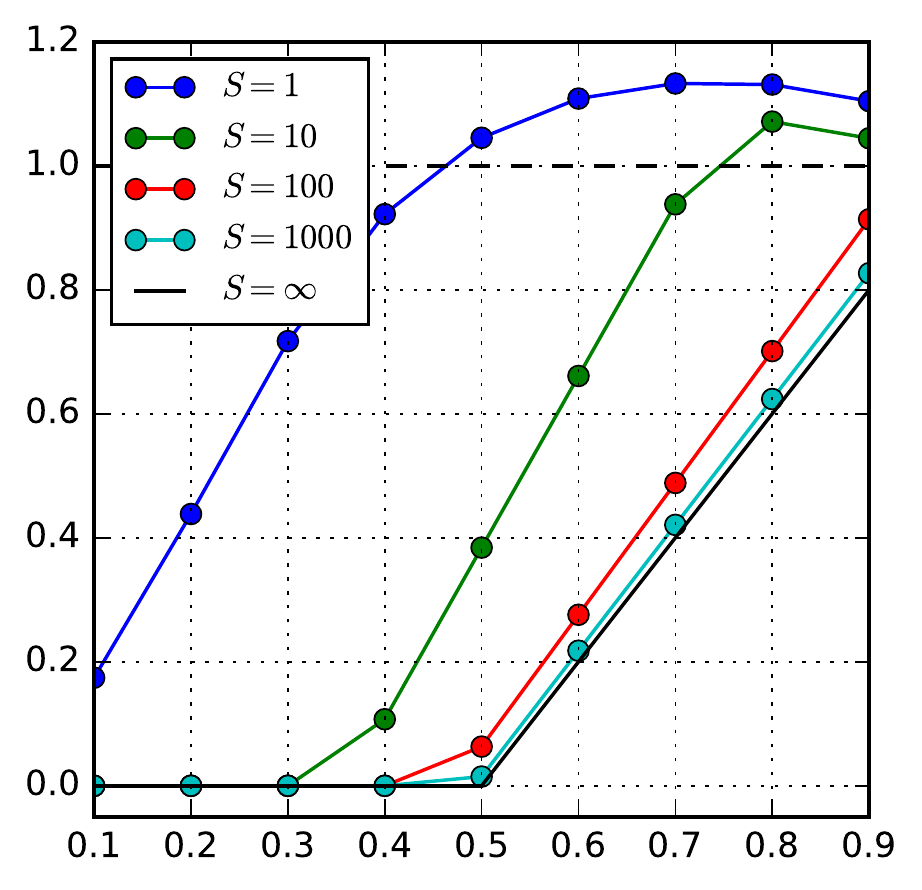}};
      \node[below=of img, node distance=0cm, yshift=3.5em,font=\color{black}] {\tiny$\ax$};
      \node[left=of img, node distance=0cm, rotate=90, anchor=center,yshift=-3em,font=\color{black}] {\tiny$\mpxaun/\mpxbun$};
      \node (ne) at (1.3, -0.6) [font=\color{black}]{\tiny $\frac{\ax-\pb}{\pa}$};
     \end{tikzpicture}
    \caption{$\mpxaun / \mpxbun$ for $\ca/\cb=1.5$}
    \label{fig:characterization b}
  \end{subfigure}
  \begin{subfigure}{.3\textwidth}
    \centering
    \begin{tikzpicture}
      \node (img)  {\includegraphics[width=.9\linewidth]{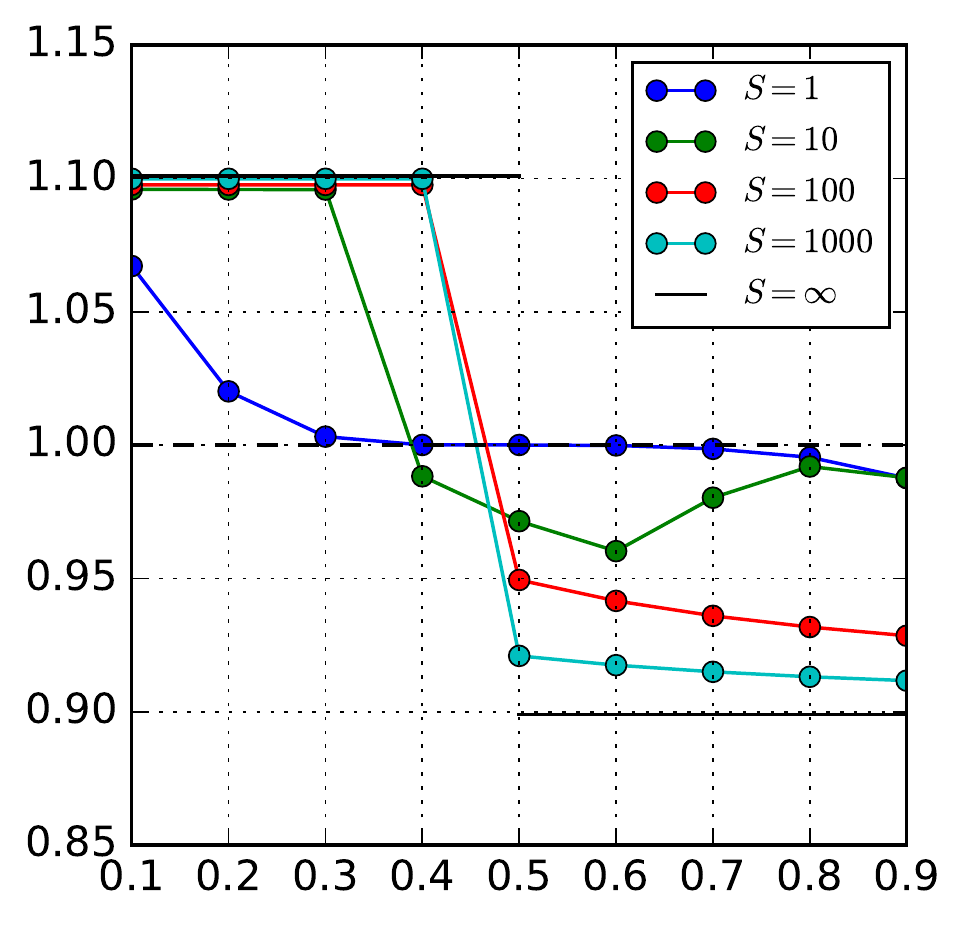}};
      \node[below=of img, node distance=0cm, yshift=3.5em,font=\color{black}] {\tiny$\ax$};
      \node[left=of img, node distance=0cm, rotate=90, anchor=center,yshift=-2.5em,font=\color{black}] {\tiny$\rev/\revdp$};
      \node (low) at (-0.4, -1.1) [font=\color{black}]{\tiny $\frac{c}{c\pa + \pb}$};
      \node (up) at (0.15, 1.35) [font=\color{black}]{\tiny $\frac{1}{c\pa + \pb}$};
     \end{tikzpicture}
    \caption{$\rev/\revdp$ for $\ca/\cb=1.5$}
    \label{fig:characterization c}
  \end{subfigure}
  \caption{Characterization of the equilibrium for different rewards $\s$. The parameters of simulations are $\cb=1$, $\pb=0.5$, $\sta=0.6$ and $\stb=1$.}
  \label{fig:characterization}
  \end{figure}

%% file: sec/discussion.tex
%!TEX root = main.tex

In this work, we propose a simple model of selection with strategic candidates who are faced with group-dependent cost-of-effort and group-dependent noise variance. We characterize the resulting discrimination at equilibrium as well as the impact of removing it through the demographic parity mechanism that mandates equal representation across groups. Note that, in the context of our strategic model, demographic parity is not the only fairness notion that might make sense, and one could be tempted to consider, for instance, a meritocratic notion of fairness (where the representation would be commensurate with the effort). However, it remains important to understand the impact of imposing demographic parity as it is one of the most commonly used fairness notions.  
% We show that, in the presence of group-dependent noise, the high-variance  candidates make a larger effort and are overrepresented. We also show that for group-dependent costs, the variance of noise does not play a role for large rewards $\s$, so only the difference in costs is important.

Throughout the paper, we made several simplifying assumptions, often to make the results easier to state and understand and to better isolate the effect of strategic behavior. Our work can be extended, however, in multiple ways:

\paragraph{Group-dependent variance of quality} 
We assumed that the variance of latent qualities is group-independent, i.e., $\sq^2_\A=\sq^2_\B = \sq^2$. This assumption can easily be removed, as all the results can equivalently be stated for $\sta^2$ and  $\stb^2$ even with group-dependent variance. We make this assumption only for simplicity of exposition since it implies that $\sta^2 < \stb^2$ if and only if $\sxa^2>\sxb^2$.

\paragraph{Unobservable effort} 
In our model, we assume that the effort $\m$ is observable to the decision-maker. If the effort $\m$ is not observable, we argue that the decision-maker performs the selection based on the noisy estimate $\wh$ rather than the posterior expectation $\wt$ (which corresponds to the group-oblivious decision-maker in the terminology of \cite{emelianov22}). All the results from this paper still hold but we need to replace the variance of the expected qualities $\stp^2$ by the variance of the estimate $\sqp^2+\sxp^2$. Note that most of the statements will be reversed as $\sta^2 <\stb^2$ if and only if $\sxa^2>\sxb^2$. In this case, the high-noise group has a higher variance of estimate, and the low-noise group has a lower variance of estimate. Hence, when the low-noise group is underrepresented in our model, the low-noise group is overrepresented in the model with unobservable effort.

\paragraph{More than two groups of candidates.}
In our model, we assume two groups of candidates, yet the results can be extended to more than two groups (e.g., in the proof of uniqueness of the equilibrium we do not rely on the fact that the number of populations is two). This additional dimension would simply add more interactions between the groups and complicate the statement of the results. For instance, in the case of multiple groups where two of them are subject to equal cost coefficients, we expect that the sorting will be with respect to the noise for these two groups, and with respect to the cost coefficients for the rest of the groups.

\paragraph{Monomial cost function and non-Gaussian noise.}
In our model, we assume a quadratic cost and Gaussian noise. We believe that these results will not change much if we assume other symmetric noise and monomial  cost functions, i.e., $\cp\m^{d_\g}$. In this case, we expect that the best response could be characterized by a dropout threshold, as in our model. In addition, the power $d_\g$ of the monomial would be another important feature of the model: if it is group-dependent, then we expect that the dropout threshold will grow as $(\s/\cp)^{1/d_\g}$ and the candidates with higher $d_\g$ will drop out earlier, independently on the relations of  $\cp$ and $\stp^2$.

\paragraph{Other models of candidate's utility}
In our model, we assume rational risk-neutral candidates faced with different costs of efforts, and we also assume a risk-neutral Bayesian decision-maker. In non-strategic settings (see e.g., \cite{emelianov22}), such a decision-maker is proven to be Bayesian-optimal, yet we prove that it is not optimal in the strategic setting. In our model, we do not consider other barriers for candidates as, for example, self-selection. To model self-selection, we can assume that there is a minimal threshold $\til_\g^{\mathrm{self}}$ that each candidate should pass. The outcome of the equilibrium would then depend on the relation between $\til_\g^{\mathrm{self}}$ and the dropout threshold $\tildi$. We can also easily consider a risk-averse (or risk loving) decision maker: in the case of an exponential utility function, this will lead to an additive bias $f_\g$ proportional to the variance $\stp^2$, i.e., 
$\wt_{i} \sim \norm(\m_{\g_i} +f_{\g_i}, \st_{\g_i}^2).$

%% file: sec/proofs.tex
\subsection{Properties of the best response}
\label{proof:br properties}
 
Recall that the payoff function of an individual with effort $\m$ and given the  selection threshold $\til$ is 
\begin{align*}
\vv(\m;\til) = \s\cdot\Phi\left(\frac{\m - \til}{\st}\right) - \frac12\cc\m^2,
\end{align*}
where $\Phi$ is the CDF of the standard normal distribution. 

Denoting $\phi$ the PDF of the standard normal distribution, the first two derivatives of $u$ with respect to $m$ are:
\begin{align*}
    \frac{\partial\vv}{\partial\m}(\m;\til)&=\frac{\s}{\st}\phi\left(\frac{\til -\m}{\st}\right) - \cc \m,\\
    \frac{\partial^2\vv}{(\partial\m)^2}(\m;\til) &= \frac{S}{\st}\phi\left(\frac{\til -\m}{\st}\right)\frac{\til -\m}{\st^2} - \cc.
\end{align*}

The payoff function $\vv$ is defined on $[0,\infty)$, and it is continuous and continuously differentiable. Moreover, $\frac{\partial\vv}{\partial\m}(\m=0) > 0 $ and $\vv(\m;\til)\xrightarrow{\m\to\infty}-\infty$. Hence, all local maxima of $\vv$ must satisfy the first-order condition (FOC) $\frac{\partial\vv}{\partial\m}(m; \til)=0$ and the second-order condition (SOC) $\frac{\partial^2\vv}{\partial\m^2}\le0$. The maximum of the payoff function is attained in one of the local maxima.

% the payoff value given fixed $\m$ and fixed $\til$. By $\vv^\til$ we denote the optimal payoff given the threshold $\til$, i.e., $\vv^\til = \vv(\mqbr(\til),\til)$. Finally, by $\vv^\mq$ we denote the payoff given of some effort assignment given $\til$, or $\vv^\mq = \vv(\mq(\til),\til)$.

% \begin{proof}
% By taking the derivative of both sides of the first-order condition and rearranging the terms, we obtain:
% \begin{align*}
%     \frac{d\m}{d\til} 
%     = \frac{\s}{\cc\st}\left(\frac{\m-\til}{\st}\right)\phi\left(\frac{\til-\m}{\st}\right)\left(\frac{1-\frac{d\m}{d\til}}{\st}\right)
%     =\frac{\m-\til}{\st^2}\m\left(1-\frac{d\m}{d\til}\right).
% \end{align*}
% \end{proof}

We start by a first lemma. 
  \begin{lemma}[Maxima of $\vv$]
    \label{lemma:maxima of payoff}
    Fix $\s$, $\cc$ and $\st$.
    \begin{enumerate}[(i)]
          \item If $\s<\cc\st^2/\phi(1)$, then there is a unique global maximum of $\vv(\m;\til)$ for all $\til$.
          \item If $\s\ge\cc\st^2/\phi(1)$, then there exists a unique $\tild(\s)$ such that for $\til=\tild(\s)$ there are two global maxima of $\vv(\m;\til)$, and for $\til\not=\tild$, there is a unique global maximum of $\vv(\m;\til)$.
      \end{enumerate}
  \end{lemma}

  \begin{proof}
    Let us denote by $z=(m-\theta)/\st$ and let $v(z)=\frac{\s}{\st}\phi(z) - C \st z$ and $w(z) := -\frac{\s}{\st}z\phi(z) - C\st = dv(z)/dz$.  The first and second derivatives of $\vv$ can be expressed as a function of $v$ and $w$:
    \begin{align*}
        &\frac{\partial\vv}{\partial\m}(\m;\til) = v(z) - C\theta,\\
        &\frac{\partial^2\vv}{(\partial\m)^2}(\m;\til) = \frac{1}{\st}w(z).
    \end{align*}
    The function $z\cdot\phi(z)$ has the global maximum at $z=1$ which is equal to $\phi(1)$, and the global minimum at $z=-1$ which is equal to $-\phi(1)$. Hence, two cases are possible:
    \begin{enumerate}[(i)]
        \item If $\phi(1) < \cc\st^2/\s$, then $w(z)<0$ for all $z\ne1$.  Hence, $v$ is a strictly decreasing function (since $w<0$ in this case), so the FOC gives a unique solution $\m$ which is a global maximum of $\vv$. 
        \item If $\phi(1)\ge\cc\st^2/\s$, then the equation $w(z)=0$ has two real solutions, denoted by $z_1 \le z_2$. They can be found explicitly, i.e., $z_{1,2} = -\sqrt{-\mathcal{W}_{1,2}\left(-\frac{2\pi C^2\st^4}{S^2}\right)}$ where $\mathcal{W}$ is the Lambert function defined as the inverse to the function $f(y)=ye^y$. Note also that $z_1\xrightarrow{\s\to\infty}-\infty$ and $z_2\xrightarrow{\s\to\infty}0$.
    \end{enumerate}

    \begin{minipage}{.5\linewidth}
        We consider the latter case (ii) in details. We can verify, that the function $v$ is a decreasing function for $z \in (-\infty, z_1)\cup (z_2,\infty)$, and it is an increasing function for $z \in (z_1, z_2)$. This shows that the function $v(z)$ has the same shape as the curve on the right. As a result, the FOC condition $v(z) - C\theta=0$ can have at most three solutions depending on the value of $\til$.
    \end{minipage}
    \begin{minipage}{.4\linewidth}
        \includegraphics[width=\linewidth]{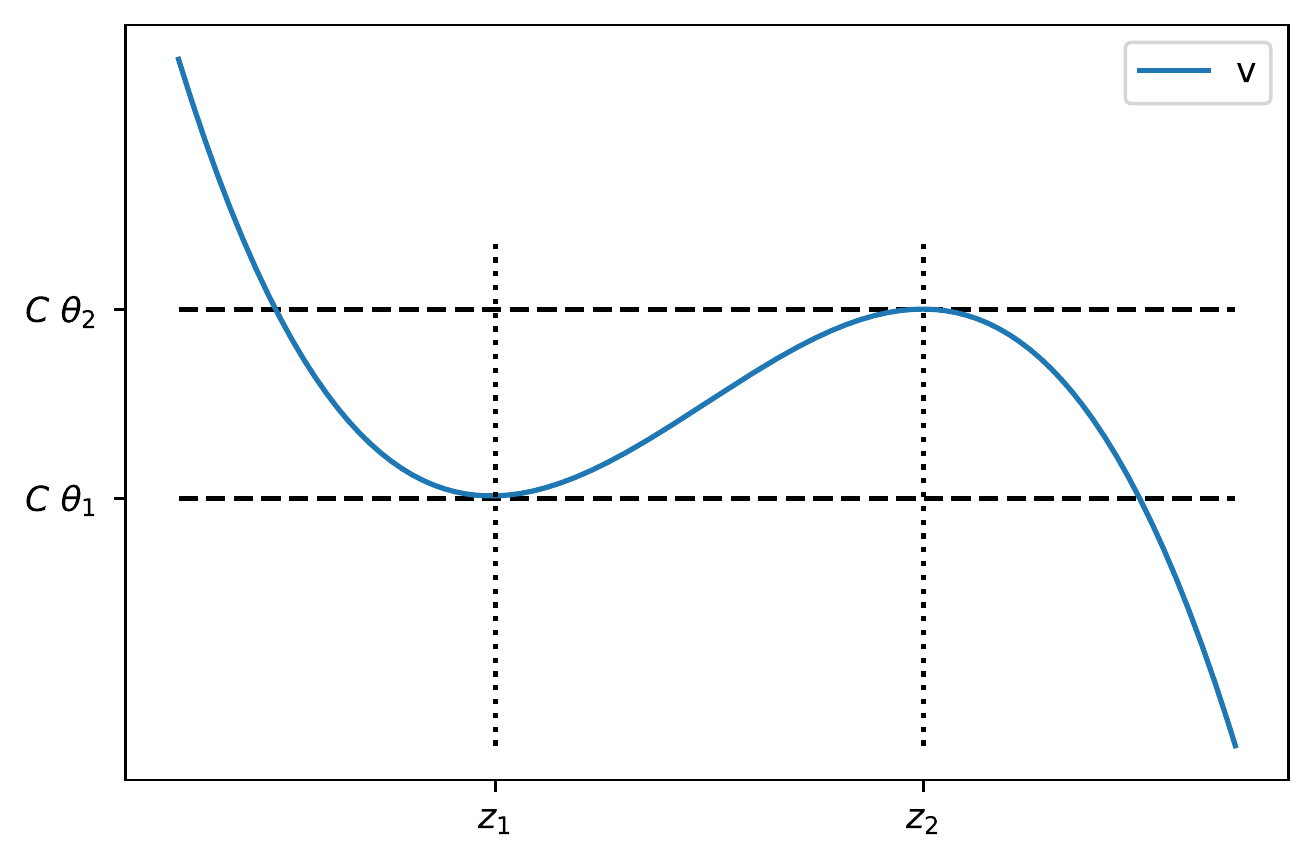}
    \end{minipage}

    Indeed, let $\theta_1=v(z_1)/C$ and $\theta_2=v(z_2)/C$. 
    \begin{itemize}
        \item For all $\til \not \in (\til_1, \til_2)$, the FOC gives a unique solution which is a global maximum of $\vv(\m;\til)$.
        \item For all $\til \in (\til_1, \til_2)$, we have that $v(z_1) \le \cc\til$ and $v(z_2)  \ge \cc\til$ which guarantees three solutions to FOC which we denote by $\m_1 \le \m_2 \le \m_3$.
        \item For $\til=\til_1$ ,we have $\vv(\m_1(\til);\til) = \vv(\m_2(\til); \til)$, and for $\til=\til_2$ we have that $\vv(\m_2(\til); \til) = \vv(\m_3(\til); \til)$. For all $\til \in (\til_1, \til_2)$, the $\m_2$ is a local minimum, $\m_1$ and $\m_3$ are local maxima.
    \end{itemize}

    Consider the two continuous, differentiable and non-negative functions $\Delta_{12}(\til) = \vv(\m_1(\til);\til) - \vv(\m_2(\til);\til)$ and  $\Delta_{32}(\til) = \vv(\m_3(\til);\til) - \vv(\m_2(\til);\til)$. We can see that $\Delta_{12}(\til_1)=0$ and $\Delta_{32}(\til_2)=0$. 
    
    We  can also verify that $\frac{\partial\vv(m(\til), \til)}{\partial\til} = \cc\m\left(\frac{\partial\m}{\partial \theta} - 1\right) - \cc \m \frac{\partial\m}{\partial \theta}= -C\m(\til)<0$. Therefore, the function $\Delta_{12}$ is increasing , since $\frac{d\Delta_{12}}{d \til} = C(m_2(\til) - \m_1(\til)) > 0$, whereas the function $\Delta_{32}$ is decreasing since $\frac{d\Delta_{32}}{d \til} = C(m_2(\til) - \m_3(\til)) < 0$. Hence, there is a unique $\tild\in(\til_1,\til_2)$, such that $\Delta_{12}(\tild) = \Delta_{23}(\tild)$ which is equivalent to $\vv(\m_1(\tild),\tild) = \vv(\m_3(\tild),\tild)$.

  \end{proof}
% \begin{lemma}[Corollary of Lemma~\ref{lemma:foc and soc}]
%     For $\s\to\infty$, we have:
%     \begin{enumerate}[(i)]
%         % \item If $\til(\s)=\tild(\s)$, then
%         \item If $\til(\s)>\tild(\s)$, then $\br(\til(\s)) - \til(\s) \xrightarrow{\s\to\infty}+\infty$
%         \item If $\til(\s)<\tild(\s)$, then $\br(\til(\s)) - \til(\s) \xrightarrow{\s\to\infty}-\infty$
%     \end{enumerate}
% \end{lemma}

% \begin{proof}
%     Using the fact that the payoff function $\vv$ is non-negative and that the selection probability $\px$ cannot be larger than 1, we show:
%     \begin{align*}
%         0 < \vv(\br(\til);\til=\s\cdot\px(\br(\til);\til) - \frac{\cc\br^2(\til)}2 \le \s - \frac{\cc\br^2(\til)} 2 \implies \br(\til) \le \sqrt{2\s/\cc}.
%       \end{align*}
%       Second, we show that $|\til - \br(\til)|$ is not bounded as $\s\to\infty$. By assuming that $|\til - \br(\til)|<\varepsilon$ for all $\s$, we end up with the following contradictory inequality that must hold for any value of $\s$:
%       \begin{align*}
%        \phi\left(\frac \varepsilon \st\right) < \phi\left(\frac{\til-\br(\til)}{\st}\right)= \frac{\cc\st}{\s}\cdot\br(\til) \le \st\sqrt{2\cc/\s} \xrightarrow{\s \to\infty} 0,
%       \end{align*}
%       where $\phi$ is the pdf of the standard normal distribution $\norm(0,1)$.
% \end{proof}

\subsection{Proof of Lemma~\ref{lemma:br large S}}
\label{proof:br large S}

We start with the proof of the case $(i)$. According to Lemma~\ref{lemma:maxima of payoff}, there are two pure best response values, $\mmax$ and $\mmin$, that correspond to the dropout threshold $\tild$. Following the definition of the expected payoff $\vv(\m;\til)$, we can show the following upper bound on the values of pure best responses $\br(\tild)\in\{\mmax(\tild), \mmin(\tild)\}$:
\begin{align*}
  0 < \vv(\br(\tild);\tild)=\s\cdot\px(\br(\tild);\tild) - \frac{\cc\br^2(\tild)}2 \le \s - \frac{\cc\br^2(\tild)} 2.
\end{align*}
This implies that $\br(\tild) \le \sqrt{2\s/\cc}$.

Second, we show that $|\tild - \br(\tild)|$ is not bounded as $\s\to\infty$. By assuming that $|\tild- \br(\tild)|<\varepsilon$ for all $\s$, we end up with the following contradictory inequality that must hold for any value of $\s$:
\begin{align*}
 \phi\left(\frac \varepsilon \st\right) < \phi\left(\frac{\tild-\br(\tild)}{\st}\right)= \frac{\cc\st}{\s}\cdot\br(\tild) \le \st\sqrt{2\cc/\s} \xrightarrow{\s \to\infty} 0,
\end{align*}
where $\phi$ is the PDF of the standard normal distribution $\norm(0,1)$.

Since the value of $|\tild -\br(\tild)|$ is not bounded, by studying the first and the second derivatives of the payoff function (as in Lemma~\ref{lemma:maxima of payoff}), we can show that $(\mmax(\tild)-\tild)\xrightarrow{\s\to\infty}+\infty$ and $(\mmin(\tild)-\tild)\xrightarrow{\s\to\infty}-\infty$. Hence, the selection rates, corresponding to the $\mmax$ and $\mmin$ converge to: 
\begin{align*}
    \lims\px(\mmax(\tild);\tild)=\lims\Phi\left(\frac{\mmax(\tild) - \tild}{\st}\right)=1,\\
    \lims\px(\mmin(\tild);\tild)=\lims\Phi\left(\frac{\mmin(\tild) - \tild}{\st}\right)=0,
\end{align*}
where $\Phi$ is the CDF of the standard normal distribution $\norm(0,1)$.

\emph{\textbf{Asymptotic behavior of $\mmin$ and $\mmax$.}}
Using the definition of the dropout threshold $\tild$ and the definition for $\mmin$, $\mmax$, we can write:
 \begin{align*}
 &0  < \vv(\mmax(\tild);\tild) =\vv(\mmin(\tild);\tild) \iff\\
 &0 < \px(\mmax(\tild);\tild) - \frac{\cc(\mmax(\tild))^2}{2\s} = \px(\mmin(\tild);\tild)- \frac{\cc(\mmin(\tild))^2}{2\s} < \px(\mmin(\tild);\tild).\\
 \end{align*}
 As $\lims\px(\mmin(\tild);\tild)=0$ and $\lims\px(\mmax(\tild);\tild)=1$, it implies that
 \begin{align*}
    \lims\frac{\cc(\mmax(\tild))^2}{2\s} = 1 \text{ and }\lims\frac{\cc(\mmin(\tild))^2}{2\s} = 0.
 \end{align*}
Hence, we show the asymptotic behavior of the pure strategy best response at $\tild$ for $\s\to\infty$:
\[
    \mmax(\tild(\s)) = \sqrt{2\s/\cc}(1+o(1))\text{ and } \mmin(\tild(\s))=o(\sqrt{2\s/\cc}).
\]

\emph{\textbf{Asymptotic behavior of $\tild$.}} 
Finally, we consider the asymptotic behavior of the dropout threshold $\tild$. The dropout threshold $\tild(\s)$ is unbounded: if we assume the opposite, then $\mmax$ would be $o(1)$, $\s\to\infty$ since $\mmax$ must satisfy the FOC:
\begin{align*}
    \mmax(\tild) = \frac{\s}{\cc\st}\phi\left(\frac{\tild - \mmax(\tild)}{\st}\right).
\end{align*}
Hence, consider the following limit which we find using l'Hôpital rule and the properties of $\mmin$ and $\mmax$ proven above:
\begin{align*}
    \lims \frac{\tild(\s)}{\sqrt{2\s/\cc}} = \lims \frac{\frac{\mmax(\tild(\s))+\mmin(\tild(\s))}{2\s}}{\frac 1 2 \s^{-1/2}\sqrt{2/\cc}} = \lims \frac{\mmax(\tild(\s))+\mmin(\tild(\s))}{\sqrt{2\s/\cc}} = 1.
\end{align*}

\emph{\textbf{Refined asymptotic behavior of $\mmin$.}} 
% We can refine the asymptotical behavior for $\mmin$.
Since $\mmin(\tild(\s)) = o(\sqrt{2\s/\cc})$ and $\tild(\s) = \sqrt{2\s /\cc}(1 + o(1))$ as $\s\to\infty$, then, using the first-order condition, we show
\begin{align*}
    \mmin(\tild(\s)) = \frac{\s}{\cc\st}\phi\left(\frac{\tild(\s)-\mmin(\tild(\s))}{\st}\right)  = \frac{\s}{\cc\st}\phi\left(\frac{\sqrt{2\s/\cc}(1+o(1))}{\st}\right)\xrightarrow{\s\to\infty}0.
\end{align*}

Now, we are ready to proof the case $(ii)$. According to Lemma~\ref{lemma:maxima of payoff}, the pure best response is unique. Consider two cases:
\begin{itemize}
\item If $\til(\s) = \gamma\tild(\s)$, then $\lims \left(\br(\gamma\tildg(\s))-\gamma\tildg(\s)\right)=+\infty$.  
Consider the following limit that we calculate using l'Hôpital rule:
\begin{align*}
    \lims \frac{\br(\gamma\tild(\s))}{\gamma\tild(\s)} = \lims \frac{d\br/d\til \cdot d\til/d\s}{d\til/d\s} = \lims \frac{d\br}{d\til}=\lims \frac{\br(\br-\til)}{\br(\br-\til)+\st^2} = 1.
\end{align*}
\item If $\til(\s) = \tild(\s)/\gamma$, then $\lims \left(\br(\tild(\s)/\gamma)-\tild(\s)/\gamma\right)=-\infty$. Since the payoff function $\vv$ is non-negative, we have:
\begin{align*}
    \vv \ge 0 \iff \cc\br(\tild(\s)/\gamma)^2/(2\s) \le  \Phi\left(\frac{\br(\tild(\s)/\gamma)-\tild(\s)/\gamma}{\st}\right) \xrightarrow{\s\to\infty}0.
\end{align*}
Hence, $\br(\tild(\s)/\gamma) = o(\sqrt{2\s/\cc})$. Given that the best response must satisfy the FOC:
\begin{align*}
    \br(\tild(\s)/\gamma) = \frac{\s}{\cc\st}\phi\left(\frac{\tild(\s)/\gamma - \br(\tild(\s)/\gamma)}{\st}\right) = \frac{\s}{\cc\st}\phi\left(\frac{\sqrt{2\s/\cc}\frac 1 \gamma(1+o(1))}{\st}\right) = o(1).
\end{align*}
\end{itemize}

\subsection{Proof of Lemma~\ref{lemma:dropout}}
\label{proof:dropout at infinity}

\textit{Case (i).} The asymptotic behavior of the dropout threshold is given in Appendix~\ref{proof:br large S}. We recall the proof here. Consider the following limit which we find using l'Hôpital rule and the properties of $\mmin$ and $\mmax$ proven in Lemma~\ref{lemma:br large S}:
\begin{align*}
    \lims \frac{\tild(\s)}{\sqrt{2\s/\cc}} = \lims \frac{\frac{\mmax(\tild(\s))+\mmin(\tild(\s))}{2\s}}{\frac 1 2 \s^{-1/2}\sqrt{2/\cc}} = \lims \frac{\mmax(\tild(\s))+\mmin(\tild(\s))}{\sqrt{2\s/\cc}} = 1.
\end{align*}

\textit{Case (ii).} To prove that the dropout threshold $\tild(\s;\st)$ is decreasing with $\st$, we differentiate with respect to $\st$ the condition on the equal payoffs for $\mmin$ and $\mmax$ strategies, and we obtain:
  \begin{align*}
    \frac{\partial\tild}{\partial\st} &= -\frac{\mmax(\tild(\s)) + \mmin(\tild(\s)) - \tild(\s)}{\st} < 0 \;\;\text{for large enough $\s$.}
  \end{align*}

\subsection{Proof of Theorem~\ref{theorem:unique equilibrium}}
\label{proof:unique equilibrium}

Let us denote by $\tiln$  the set-valued function
\begin{align*}
    \tiln(\til) = \{F^{-1}_{\mvec}(1-\ax): \mup \in \BR_\g(\til)\}, 
\end{align*}
where $\mvec = (\mqa, \mqb)$ and $\F_{\mvec}$ is the CDF of $\wt$ induced by $\mvec$: $\F_{\mvec} = \pa\F_{\mqa} + \pb\F_{\mqb}$.

\begin{lemma}
    \label{lemma:equivalence}
    There is a one-to-one correspondence between the fixed points of $\tiln$ and the equilibria $\mvecun$ of the game $\gun$. 
\end{lemma}

\begin{proof}
    If $\mvecun$ is an equilibrium of the game $\gun$, then there is a unique $\tilun$ such that $\F^{-1}_{\mvecun}(1-\ax)=\tilun$ since $\F_{\mvecun}(\til)$ is a monotone function. This $\tilun$ is a fixed point of $\tiln$ since  $\mqgun \in \BR_\g(\tilun)$ by definition of $\mvecun$.

    On the other hand, if $\tilun$ is a fixed point of $\tiln$, then there is a unique $\mvecun=(\mqaun,\mqbun)$ such that $\mqgun\in\BR_\g(\tilun)$: 
    \begin{enumerate}
        \item If $\til\not=\tildg$ for all $\g\in\{\A,\B\}$, then the pure best responses $\br_\A$ and $\br_\B$  are unique (see Lemma~\ref{lemma:maxima of payoff} in Appendix~\ref{proof:br properties}), so $\mvecun$ is unique.
        \item If, without loss of generality, $\til=\tilda$, then there is a unique pure best response $\br_{\B}$, and two pure response values $\mqamax$, $\mqamin$ (see Lemma~\ref{lemma:maxima of payoff} in Appendix~\ref{proof:br properties}). The function $\F(\tau)=\F_{\mvec}(\til)$ is monotone in $\tau\in[0,1]$ which is a probability for $\mqa=\tau\delta(\m-\mqamax) +(1-\tau)\delta(\m-\mqamin)$. It is also an equilibrium of $\gun$ as $\til = \F^{-1}_{\mvec}(1-\ax)$ by definition of $\til$. Hence, $\mvecun=(\mqa, \mqb)$ is an equilibrium of the game $\gun$ and it is unique.
    \end{enumerate}
\end{proof}

By showing a one-to-one correspondence between equilibria of $\gun$ and the fixed points of $\tiln$, it is left us to prove that the function $\tiln(\til)$ has a unique fixed point $\tilun$, i.e., a solution to $\tiln(\tilun) = \tilun$ is unique. This will imply that the equilibrium of the game $\gun$ is unique. We use the following lemmas.

\begin{lemma}
    \label{lemma:derivative br}
        Assume that $\m(\til)$ satisfies FOC and SOC defined in Section~\ref{proof:br properties}.
        \begin{enumerate}[(i)]
            \item If $\m > \til$, then  $0< \frac{d\m}{d\til} < 1$.
            \item If $\m \le \til$, then  $\frac{dm}{d\til} \le 0$.
        \end{enumerate}
    \end{lemma}
    
    \begin{proof}
    First, we derive the expression for the first derivative. We differentiate the both sides of FOC defined in Appendix~\ref{proof:br properties}:
        \begin{align*}
            \frac{d\m}{d\til} 
            = \frac{\s}{\cc\st}\left(\frac{\m-\til}{\st}\right)\phi\left(\frac{\til-\m}{\st}\right)\left(\frac{1-\frac{d\m}{d\til}}{\st}\right)
            =\frac{\m-\til}{\st^2}\m\left(1-\frac{d\m}{d\til}\right) 
            \implies \frac{d \m}{d\til} = \frac{\m(\m-\til)}{\m(\m-\til)+\st^2}.
        \end{align*}
    
    Second, from the SOC defined in Appendix~\ref{proof:br properties}, we find that $\m\left(\m - \til\right) + \st^2 > 0$. Hence, the sign of the derivative ${d\m}/{d\til}$ is determined only be the sign of its numerator $\m(\m - \til)$. Note that the value of  $\m$ is strictly positive as it satisfies the FOC, so ${d\m}/{d\til} >0$ if and only if $\m > \til$.
    \end{proof}

\begin{lemma}
    \label{lemma:total derivative theta}
    For all $\s$ and $\til\not=\tildi(\s)$, the total derivative $d\tiln/d\til$ can be expressed as:
    $$\frac{d \tiln}{d\til} =  \frac{\sum_\g \p_\g \frac 1 \stp \phi\left(\frac{\tiln  - \br_\g}{ \stp}\right)\cdot \frac{\br_\g(\br_\g-\til)}{\br_\g(\br_\g-\til)+\stp^2}}{\sum_\g \p_\g \frac 1 \stp \phi\left(\frac{\tiln - \br_\g}{\stp}\right)} < 1.$$
\end{lemma}

\begin{proof}
    Since $\tiln$ is an implicit function of $\til$, we use the chain rule to find the total derivative:
    \begin{align*}
        \frac{d\tiln}{d\til} = \sum_\g\frac{\partial \tiln}{\partial\br_\g}\frac{d\br_\g}{d\til}.
    \end{align*}
    By differentiating both sides of the budget constraint, we can find the partial derivative ${\partial \tiln}/{\partial\m_\g}$: 
    \begin{align*}
        \frac{\partial}{\partial\br_\g}\left(\sum_\g \p_\g \pxg(\br_\g;\tiln)\right)  = 0 \iff \frac{\partial \tiln}{\partial\br_\g}  = \frac{\p_\g \frac 1 \stp \phi\left(\frac{\tiln  - \br_\g}{ \stp}\right)}{\sum_\g \p_\g \frac 1 \stp \phi\left(\frac{\tiln  - \br_\g}{ \stp}\right)}.
    \end{align*}

    Finally, using the fact that ${\partial \tiln}/{\partial \br_\g} > 0$, where $\sum_\g {\partial \tiln}/{\partial \br_\g}=1$ and Lemma~\ref{lemma:derivative br}, we show that ${d\tiln}/{d\til} < 1$.
    
\end{proof}

% Now we are ready to prove the existence and the uniqueness of the equilibrium. First, let us show that there is a $\til_0$, such that $\tiln(\til_0) > \til_0$. For any fixed $\ax$, let $\til_0$ be the solution to the equation $\sum_\g \p_\g\Phi^c\left(\til_0/\stp\right) = \ax$. Since the best response $\br_\g(\til_0) > 0$ for all $\g\in\{\A,\B\}$, then $\tiln(\til_0) > \til_0$.

% We prove that for some $\til_0$, we have that $\tiln(\til_0)>\til_0$. We also show that the function $\tiln$ is continuous and $d\tiln/d\til<1$ for all $\til\not=\tildi$, the function $\tiln$ is upper hemicontinuous for all $\til=\tildg$ for all $\g\in\{\A,\B\}$. Hence, there is a unique solution to the fixed-point problem $\tiln(\til)=\til$, and, as a result, a  unique equilibrium of the game $\gun$ due to Lemma~\ref{lemma:equivalence}.

\paragraph{Existence}
First, let us show that there is an interval $[\til_0, \til_1]$, such that for all $\til \in\R$, we have that $\til_0 \le \tiln(\til) \le \til_1$:
\begin{itemize}
    \item Since the best response $\br_\g(\til) \ge 0$ for all $\g\in\{\A,\B\}$, then for any fixed $\ax$, let $\til_0$ be the solution to the equation $\sum_\g \p_\g\Phi^c\left(\frac{\til_0-0}{\stp}\right) = \ax$. Hence, $\tiln(\til) \ge \til_0$ for all $\til\in\R$.
    \item Since the best response $\br_\g(\til) \le \sqrt{2\s/\cp}$ for all $\g\in\{\A,\B\}$, then for any fixed $\ax$, let $\til_1$ be the solution to the equation $\sum_\g \p_\g\Phi^c\left(\frac{\til_1-\sqrt{2\s/\cp}}{\stp}\right) = \ax$. Hence, $\tiln(\til) \le \til_1$  for all $\til\in\R$.
\end{itemize}
Therefore, we can consider the function $\tiln$ on the interval $[\til_0, \til_1]$ which is compact and convex. The graph of the function $\tiln$ is closed and for all $\til\in[\til_0, \til_1]$, we have that $\tiln(\til)$ is convex (for $\til\not=\tildg$, the value of $\tiln(\theta)$ is unique, and for $\til=\tildg$, the value of $\tiln(\theta)$ is an interval). Hence, using Kakutani fixed point theorem, we show that the fixed point exists.

\paragraph{Uniqueness}
We show that there exists a fixed point of the function $\tiln$.  For all $\til\not=\tildi$, we have that $d\tiln/d\til<1$ (Lemma~\ref{lemma:total derivative theta}), and we have that $\lim_{\til\to\tildg+} \tiln(\til)\le \lim_{\til\to\tildg-} \tiln(\til)$ for all $\g\in\{\A,\B\}$. Hence, there is a unique solution to the fixed-point problem $\tiln(\til)=\til$, and, as a result, a  unique equilibrium of the game $\gun$ due to Lemma~\ref{lemma:equivalence}.

\subsection{Proof of Theorem~\ref{theorem:equilibrium strategy}}
\label{proof:equilibrium strategy}

\subsubsection*{Case (i)}

We prove that the dropout threshold $\tild_1(\s)$ is the fixed point of $\tiln$, i.e., it corresponds to the equilibrium of the game $\gun$. As we show in the proof of Lemma~\ref{lemma:br large S}, as $\s\to\infty$, we have $\pxmax_\g:= \px(\mmax(\tildg); \tildg) \xrightarrow{\s\to\infty}1$ and $\pxmin_\g:=\px(\mmin(\tildg); \tildg)\xrightarrow{\s\to\infty}0$. 
% We also prove that for $\ca=\cb$ and $\sta^2<\stb^2$, we have $\tilda > \tildb$ for substantially large $S$. 

If $\ax \le p_1$, and given that the selection rate for the $G_2$-candidates at $\tild_1$ tends to zero with $\s$, assume that $G_1$-candidates randomize their strategy by playing $\mmax_1$ with the probability $\tau_1$ and $\mmin_1$ with the probability $1-\tau_1$. The $G_2$-candidates play $\mqbr_2(\tild_1)\xrightarrow{\s\to\infty}0$. The probability $\tau_1$ can be found from the budget constraint: $\ax = p_1(\tau_1\px_1^{\max} + (1-\tau_1)\px_1^{\min}) + p_1 \px_2(\tild_1,\mqbr_2(\tild_1)) \xrightarrow{\s\to\infty} p_1\tau_1$, so $\tau_1\xrightarrow{\s\to\infty}\ax/p_1$. This strategy satisfies the budget constraint, so $\tild_1$ is the fixed point of $\tiln$.

\subsubsection*{Case (ii)}

We now prove that the dropout threshold $\tild_2$ is fixed point of $\tiln$. If $\ax > p_1$, then selecting all candidates from $G_1$ group would not be enough, and some $G_2$-candidates are needed to satisfy the selection rate $\ax$.

Given $\tild_2$, the fraction of selected $G_1$-candidates would be equal to $\px_1(\tild_2, \mqbr_1(\tild_2))\xrightarrow{\s\to\infty}1$. We claim the $G_2$-candidates will play $\mmax_2$ with probability $\tau_2$, and $\mmin_2$ with probability $1-\tau_2$. The probability $\tau_2$ can be found from the budget constraint: 
\begin{align*}
    \ax = p_1 \px_1(\tild_2, \mqbr_1(\tild_2)) + p_2 (\px_2^{\max}\tau_2 + \px_2^{\min}(1-\tau_2)) \xrightarrow{\s\to\infty} p_1 + p_2 \tau_2.
\end{align*}
Hence, for $\tau_2 \xrightarrow{\s\to\infty} \frac{\ax -p_1}{p_2}$, the dropout threshold, $\tild_2$ is the fixed point of $\tiln$.

\subsection{Proof of Corollary~\ref{theorem:game comparison}}
\label{proof:game comparison}

We use the expressions for equilibrium strategies from Theorem~\ref{theorem:equilibrium strategy} and Theorem~\ref{theorem:dp strategy}.

(i) If $\ca=\cb=\cc$ and $\sxa^2>\sxb^2$:
\begin{align*}
    \lims \frac \mmqaun \mmqadp = \lims 
    \begin{cases}
        \frac{\ax/\pa\sqrt{2\s/\cc}(1+o(1))}{\ax\sqrt{2\s/\cc}(1+o(1))} &\text{ if } \ax \le \pa\\
        \frac{\sqrt{2\s/\cc}(1+o(1))}{\ax\sqrt{2\s/\cc}(1+o(1))} &\text{ if } \ax > \pa\\
    \end{cases}
    =
    \begin{cases}
        1/\pa &\text{ if } \ax \le \pa,\\
        1/\ax &\text{ if } \ax > \pa,\\
    \end{cases}
\end{align*}
\begin{align*}
    \lims \frac \mmqbun \mmqbdp = \lims 
    \begin{cases}
        \frac{o(1)}{\ax\sqrt{2\s/\cc}(1+o(1))} &\text{ if } \ax \le \pa\\
        \frac{\frac{\ax-\pa}{\pb}\sqrt{2\s/\cc}(1+o(1))}{\ax\sqrt{2\s/\cc}(1+o(1))} &\text{ if } \ax > \pa\\
    \end{cases}
    =
    \begin{cases}
        0 &\text{ if } \ax \le \pa,\\
        \frac{\ax-\pa}{\ax - \ax\pa} &\text{ if } \ax > \pa.\\
    \end{cases}
\end{align*}

(ii) If $\ca>\cb$:
\begin{align*}
    \lims \frac \mmqaun \mmqadp = \lims 
    \begin{cases}
        \frac{o(1)}{\ax\sqrt{2\s/\ca}(1+o(1))} &\text{ if } \ax \le \pb\\
        \frac{(\ax-\pb)/\pa\sqrt{2\s/\ca}(1+o(1))}{\ax\sqrt{2\s/\ca}(1+o(1))} &\text{ if } \ax > \pb\\
    \end{cases}
    =
    \begin{cases}
        0 &\text{ if } \ax \le \pb,\\
        \frac{\ax - \pb}{\ax - \ax\pb} &\text{ if } \ax > \pb,\\
    \end{cases}
\end{align*}
\begin{align*}
    \lims \frac \mmqbun \mmqbdp = \lims 
    \begin{cases}
        \frac{\ax/\pb\sqrt{2\s/\cb}(1+o(1))}{\ax\sqrt{2\s/\cb}(1+o(1))} &\text{ if } \ax \le \pb\\
        \frac{\sqrt{2\s/\ca}(1+o(1))}{\ax\sqrt{2\s/\cb}(1+o(1))} &\text{ if } \ax > \pb\\
    \end{cases}
    =
    \begin{cases}
        1/\pb &\text{ if } \ax \le \pb,\\
        \sqrt{\frac \cb \ca}\frac 1 \ax &\text{ if } \ax > \pb.\\
    \end{cases}
\end{align*}

\subsection{Proof of Theorem~\ref{theorem:revenue ratio}}
\label{proof: quality ratio}

First, using the formula for the expected value of the truncated normal distribution, we find:
\begin{align*}
    \rev(\m_\g;\til) &= \sum_\g \p_\g \E(\w_\g\cdot[\wt_\g > \til_\g]) = \sum_\g \p_\g\E(\w_\g | \wt_\g > \til_\g)\Pb(\wt_\g\ge\til_\g) \\
    & = \sum_\g \p_\g \left[\m_\g \Phi\left(\frac{\til_\g-\m_\g}{\stp}\right) + \stp \phi\left(\frac{\til_\g-\m_\g}{\stp}\right)\right].
\end{align*}
Second, using Lemma~\ref{lemma:br large S}, we can show that for all $\gamma\in(0,1)$ and $\tilg = \gamma\cdot\tildg$, we have:
\begin{align*}
    \br_\g(\tilg)\cdot\Phi\left(\frac{\tilg-\br_\g(\tilg)}{\stp}\right) + \stp \phi\left(\frac{\tilg-\br_\g(\tilg)}{\stp}\right)=\tilg(1+o(1))(1+o(1)) + \stp \cdot o(1) = \tilg(1+o(1))
\end{align*}
and, for $\tilg = \tildg/\gamma$, we have:
\begin{align*}
    \br_\g(\tilg)\cdot \Phi\left(\frac{\tilg-\br_\g(\tilg)}{\stp}\right) + \stp \phi\left(\frac{\tilg-\br_\g(\tilg)}{\stp}\right)\xrightarrow{\s\to\infty}0.
\end{align*}
For $\tilg = \tildg$, we have:
\begin{align*}
    &\mmax_\g(\tilg)\cdot\Phi\left(\frac{\tilg-\mmax_\g(\tilg)}{\stp}\right) + \stp \phi\left(\frac{\tilg-\mmax_\g(\tilg)}{\stp}\right) = \tilg(1+o(1)),\\
    &\mmin_\g(\tilg)\cdot \Phi\left(\frac{\tilg-\mmin_\g(\tilg)}{\stp}\right) + \stp \phi\left(\frac{\tilg-\mmin_\g(\tilg)}{\stp}\right)\xrightarrow{\s\to\infty}0.
\end{align*}

For the game $\gdp$, we use the equilibrium strategy from Theorem~\ref{theorem:dp strategy}, and can write that:
\begin{align*}
    \revdp
    &= \pa(\ax+o(1))\cdot\tilda(\s)(1+o(1)) + \pb \cdot \left(\ax+o(1)\right)\tildb(\s)(1+o(1))\\
    & = \ax\pa \cdot \tilda(\s) (1+o(1)) + \ax\pb \cdot \tildb(\s) (1+o(1)), \s\to\infty.
\end{align*}

For the game $\gun$, we use the equilibrium strategy from Theorem~\ref{theorem:equilibrium strategy}. We consider separately the case of group-independent and the case of group-dependent costs.

\paragraph{\textbf{Group-independent cost}}
\begin{align*}
\intertext{If $\ax \le \pa$, then}
    \rev &= \pa({\ax}/{\pa}+o(1))\cdot\tilda(1+o(1)) + \pb \cdot o(1) = \ax \cdot \sqrt{2\s/\cc} (1+o(1)), \s\to\infty.
\intertext{If $\ax > \pa$, then}
    \rev &= \pa(1+o(1))\cdot\tildb(1+o(1)) + \pb \cdot \left(\frac{\ax - \pa}{\pb}+o(1)\right)\tildb(1+o(1)) = \ax\sqrt{2\s/\cc}(1+o(1)), \s\to\infty,
\end{align*}
where we use the assumption on equal costs $\ca=\cb$ and the result of Lemma~\ref{lemma:dropout} on the asymptotic behavior of the dropout threshold.
Hence, we can calculate the limit:
\begin{align*}
    \lims \revdp/\rev = 1.
\end{align*}

\paragraph{\textbf{Group-dependent cost}}
\begin{align*}
\intertext{If $\ax \le \pb$, then}
    \rev &= \pb (\ax/\pb +o(1)) \cdot \tildb\cdot(1+o(1)) + \pa \cdot o(1) = \ax \cdot \sqrt{2\s/\cb}\cdot(1+o(1)).
\intertext{If $\ax > \pb$, then}
    \rev &= \pb (1 +o(1)) \cdot \tilda\cdot(1+o(1)) + \pa \cdot \left(\frac{\ax-\pb}{\pa}+o(1)\right)\tilda(1+o(1))\\
    & = \pb(1+o(1)) \tilda\cdot(1+o(1)) + (\ax -\pb)\tilda(1+o(1)).
\end{align*}
Hence, for $c:= \sqrt{\cb/\ca}$, we can write that:
\begin{align*}
    \lims \rev/\revdp = 
    \begin{cases}
        \lims \frac{\ax\tildb(1+o(1))}{\ax\pa\tilda(1+o(1)) + \ax\pb\tildb(1+o(1))} = \frac{1}{c\pa+\pb} & \text{if }\ax \le \pb,\\
        \lims \frac{\pb\tilda(1+o(1)) + (\ax -\pb)\tilda(1+o(1))}{\ax\pa\tilda(1+o(1)) + \ax\pb\tildb(1+o(1))} = \frac{c}{c\pa + \pb} & \text{if } \ax > \pb,
    \end{cases}
\end{align*}
where $\frac{c}{c\pa + \pb} < \frac{c}{c\pa + c\pb} = 1$.

\subsection{Proof of Theorem~\ref{theorem:disparity of effort}}
\label{proof:inequality of effort}

\paragraph{\textbf{Equilibrium efforts}}
Let us  first find for which $\ax\in(0,1)$ we have the equality of effort at equilibrium, i.e.,  $\maun=\mbun$. Since the equilibrium values are the best responses to $\tilun$, and the pure best response is unique (see Lemma~\ref{lemma:maxima of payoff}), we can write:
\begin{align*}
    \maun = \mbun  &\iff \frac{\s}{\ca\sta}\phi\left(\frac{\tilun - \maun}{\sta}\right) = \frac{\s}{\cb\stb}\phi\left(\frac{\tilun - \mbun}{\stb}\right) \\
    &\iff \tilun - \maun =  \pm \sqrt{-2\frac{\log(\ca\sta)-\log (\cb\stb)}{1/\sta^2 - 1/\stb^2}}.
\end{align*}

For $\ca\sta >  \cb\stb$, there exist no real solution to the above equation. In this case, we can show that $\maun < \mbun$ for all $\ax \in (0,1)$. Assume the opposite, i.e., $\maun > 
\mbun$ for all $\til$, then $\phi((\til-\maun)/\sta) < \phi((\til-\mbun)/\stb)$ and also $\ca\sta > \cb\stb$ which contradicts the initial assumption. 

For $\ca\sta \le \cb\stb$, there are two values of $\ax_{\maun=\mbun}$ which correspond to equal equilibrium efforts $\maun=\mbun$:
\begin{align*}
    \ax^{(1,2)}_{\maun=\mbun} = \sum_G \pg\Phi^c\left(\pm \frac{\sqrt{-2\frac{\log(\ca\sta)-\log(\cb\stb)}{1/\sta^2 - 1/\stb^2}}}{\stp}\right).
\end{align*}
We can verify that ${d\mqabr}/{d\ax} > {d\mqbbr}/{d\ax}$ for $\ax^{(1)}_{\maun=\mbun}$ and ${d\mqabr}/{d\ax} < {d\mqbbr}/{d\ax}$ for $\ax^{(2)}_{\maun=\mbun}$ which concludes the proof.
 
\paragraph{\textbf{Equilibrium selection rates}}
Let us find such $\ax_{\pxa = \pxb}$ for which both groups are selected at equal rates, i.e., $\pxaun = \pxbun$. This is equivalent to: 
\begin{align*}
    \Phi^c\left(\frac{\tilun-\maun}{\sta}\right) =     \Phi^c\left(\frac{\tilun-\mbun}{\stb}\right)=\ax_{\pxa = \pxb}
    \iff \tilun = \mpun + \stp\Phi^{-1}(1-\ax_{\pxa = \pxb}), \;\forall \g\in\{\A,\B\}.
\end{align*}
By definition of the best response, we can also write that:
\begin{align*}
    \mpun &= \frac{\s}{\cp\stp}\phi\left(\frac{\tilun-\mpun}{\stp}\right) = \frac{\s}{\cp\stp}\phi\left(\Phi^{-1}(1-\ax_{\pxa = \pxb})\right),\\
    \tilun &= \frac{\s}{\cp\stp}\phi\left(\Phi^{-1}(1-\ax_{\pxa = \pxb})\right)+\stp\Phi^{-1}(1-\ax_{\pxa = \pxb}).
\end{align*}
The equality for $\tilun$ is possible only for such $\ax$, when
\begin{align*}
    \frac{\s}{\ca\sta}\phi\left(\Phi^{-1}(1-\ax_{\pxa = \pxb})\right)+\sta\Phi^{-1}(1-\ax_{\pxa = \pxb})=
    \frac{\s}{\cb\stb}\phi\left(\Phi^{-1}(1-\ax_{\pxa = \pxb})\right)+\stb\Phi^{-1}(1-\ax_{\pxa = \pxb}),
\end{align*}
which is if and only if
\begin{align*}
    \s\left(\frac 1 {\ca\sta}-\frac 1 {\cb\stb}\right)\phi\left(\Phi^{-1}(1-\ax_{\pxa = \pxb})\right)= \Phi^{-1}(1-\ax_{\pxa = \pxb})(\stb -\sta).
\end{align*}
We solve the corresponding equation by letting $z\coloneqq\Phi^{-1}(1-\ax_{\pxa = \pxb})$ and solving with respect to $z$. After, $\ax_{\pxa = \pxb} = 1- \Phi(z)$. 

Let $\xi = \s\left(\frac 1 {\ca\sta}-\frac 1 {\cb\stb}\right)/(\stb-\sta)$. Then by definition of $\phi$:
\begin{align*}
    \xi\frac{1}{\sqrt{2\pi}}e^{-\frac{z^2}{2}} = z \iff z\cdot e^{\frac{z^2}{2}} = \xi\frac{1}{\sqrt{2\pi}}.
\end{align*}
By squaring both sides of the above equation, we end up with the equation of type $y\cdot e^y = const$ which can be solved using the Lambert $\mathcal{W}$ function defined as the inverse to $f(y)=y\cdot e^y$. As a result, by solving for $z^2$ and taking the square root with the sign equal to the sign of $\xi$, we show:  
\begin{align*}
    \ax_{\pxa = \pxb} = 1- \Phi\left(\sign(\xi)\sqrt{\mathcal{W}\left(\frac{\xi^2}{2\pi}\right)}\right).
\end{align*}
We can verify that the $d\pxa/d\ax > d\pxb/d\ax$ at $\ax_{\pxa = \pxb}$, so for all $\ax < \ax_{\pxa = \pxb} = 1- \Phi\left(\sign(\xi)\sqrt{W\left(\frac{\xi^2}{2\pi}\right)}\right)$, we have that $\pxaun < \pxbun$.